\documentclass[a4paper,UKenglish,cleveref, autoref,numberwithinsect]{lipics-v2021}

\usepackage[T1]{fontenc}
\usepackage{graphicx}

\usepackage[utf8]{inputenc}

\usepackage{longtable}
\usepackage{rotating}
\usepackage{textcomp}
\usepackage{version}
\usepackage{hyperref}
\usepackage{amsmath}
\usepackage{amsfonts}
\usepackage{amssymb,dsfont}
\usepackage{mathtools}

\usepackage{appendix}
\usepackage[svgnames]{xcolor}
\usepackage{tikz}
\usetikzlibrary{automata,positioning, backgrounds}
\usetikzlibrary{arrows,automata}
\usetikzlibrary{shapes}
\usetikzlibrary{fit}
\usepackage{caption}
\usepackage{subcaption}
\usepackage{ stmaryrd }
\usepackage{float}
\usepackage{comment}
\usepackage{framed}

\usepackage{tcolorbox}

\usetikzlibrary{shapes.geometric,arrows}
\usetikzlibrary{decorations.pathmorphing}
\usepackage{lineno}
\linenumbers

\usepackage{tabularx,lipsum,environ}

\newcommand{\N}{\mathbb{N}}
\newcommand{\tww}[1]{\mathrm{tww}(#1)}

\DeclareMathOperator{\tw}{tw}
\DeclareMathOperator{\twww}{tww}
\DeclareMathOperator{\dtw}{dtw}
\DeclareMathOperator{\dtww}{dtww}
\DeclareMathOperator{\ttw}{tw\ensuremath{_{\rightarrow}}}
\DeclareMathOperator{\ttwww}{tww\ensuremath{_{\rightarrow}}}
\DeclareMathOperator{\qr}{qr}
\DeclareMathOperator{\dist}{d}

\newcommand{\unionG}[1]{\ensuremath{#1_\downarrow}}
\newcommand{\expG}[1]{\ensuremath{#1_\rightarrow}}
\newcommand{\tempE}{\ensuremath{E^\textrm{temp}}}
\newcommand{\adjE}{\ensuremath{E^\textrm{loc}}}

\newcommand{\localstatic}[1]{\ensuremath{G}^r}

\newcommand{\FO}{\ensuremath{\textsf{FO}}}
\newcommand{\MSO}{\ensuremath{\textsf{MSO}}}
\newcommand{\MSOT}{\ensuremath{\textsf{MSOT}}}
\newcommand{\MSOone}{\ensuremath{\textsf{MSOT}_1}}
\newcommand{\FOone}{\ensuremath{\textsf{FOT}_1}}
\newcommand{\MSOtwo}{\ensuremath{\textsf{MSOT}_2}}

\newcommand{\inc}{\operatorname{inc}}
\newcommand{\app}{\operatorname{app}}
\newcommand{\adj}{\operatorname{adj}}
\newcommand{\suc}{\operatorname{succ}}
\newcommand{\simd}{\operatorname{sim}}
\newcommand{\ar}{\ensuremath{\textrm{ar}}}
\newcommand{\struct}[1]{\ensuremath{\langle #1 \rangle}}
\newcommand{\structone}[1]{\ensuremath{\lceil #1\rceil}}
\newcommand\Aset{\mathcal{A}}

\newcommand{\graphstruct}[1]{\lfloor #1 \rfloor}
\newcommand\Tbb{\mathbb{T}}
\newcommand{\width}{\operatorname{width}}
\newcommand{\Gset}{\mathcal{G}}

\newcommand{\closedsuc}{\ensuremath{\mathtt{ClosedSucc}}}
\newcommand{\sucstar}{\ensuremath{\mathtt{succ}^*}}
\newcommand{\closedtempsuc}{\ensuremath{\mathtt{ClosedTempSucc}}}
\newcommand{\temppath}{\ensuremath{\mathtt{TempPath}}}

\newcommand{\timevertex}{\ensuremath{\mathtt{{TimeVertex}}}}
\newcommand{\timeformula}{\ensuremath{\mathtt{{Time}}}}
\newcommand{\timeadj}{\ensuremath{\mathtt{TimeAdj}}}

\newcommand{\tempmatch}{\ensuremath{\mathtt{TempMatch}}}
\newcommand{\gammaedge}{\ensuremath{\gamma-\mathtt{Edge}}}

\newcommand{\structwo}[1]{\ensuremath{\lbrack #1 \rbrack}}
\newcommand{\Gaif}{\ensuremath{\textit{Gaif}}}
\newcommand{\neigh}[1]{\ensuremath{\mathcal{N}^{#1}}}
\newcommand{\ie}{i.e.\ }

\newcommand{\resp}{resp.\ }

\usepackage{stackengine}
\stackMath
\usepackage{ifthen}

\newcounter{toutpetit}

\title{Model checking with temporal graphs and their derivative}

\author{Binh-Minh Bui-Xuan}{LIP6, CNRS, Sorbonne Universit\'e, France}{buixuan@lip6.fr}{}{}
\author{Florent Krasnopol}{Universit\'e de Lorraine, CNRS, Inria, LORIA, France}{florent.krasnopol@loria.fr}{}{}
\author{Bruno Monasson}{ENS de Lyon, France}{bruno.monasson@ens-lyon.fr}{}{}
\author{Nathalie Sznajder}{LIP6, CNRS, Sorbonne Universit\'e, France}{nathalie.sznajder@lip6.fr}{}{}

\authorrunning{Binh-Minh Bui-Xuan, Florent Krasnopol, Bruno Monasson, and Nathalie Sznajder} 

\Copyright{Binh-Minh Bui-Xuan, Florent Krasnopol, Bruno Monasson, and Nathalie Sznajder} 



\category{} 

\relatedversion{} 

\acknowledgements{We are grateful to the anonymous reviewers of all our submissions for their helpful comments which greatly improved the paper.}

\nolinenumbers 

\hideLIPIcs  

\EventEditors{John Q. Open and Joan R. Access}
\EventNoEds{2}
\EventLongTitle{42nd Conference on Very Important Topics (CVIT 2016)}
\EventShortTitle{CVIT 2016}
\EventAcronym{CVIT}
\EventYear{2016}
\EventDate{December 24--27, 2016}
\EventLocation{Little Whinging, United Kingdom}
\EventLogo{}
\SeriesVolume{42}
\ArticleNo{23}

\begin{document}

\maketitle
\begin{abstract}
Temporal graphs are graphs where the presence or properties of their vertices and edges change over time.
When time is discrete, a temporal graph can be defined as a sequence of static graphs over a discrete time span, called lifetime, or as a single graph where each edge is associated with a specific set of time instants where the edge is alive.
For static graphs, Courcelle's Theorem asserts that any graph problem expressible in monadic second-order logic can be solved in linear time on graphs of bounded tree-width.
We propose the first adaptation of Courcelle's Theorem for monadic second-order logic on temporal graphs that does not explicitly rely on a parameter proportional to the lifetime, or defined as the maximum number of time-edges incident with any vertex which in the worst case is higher than the lifetime.
We then introduce the notion of derivative over a sliding time window of a chosen size, and define the tree-width and twin-width of the temporal graph's derivative.
We exemplify its usefulness with meta-theorems with respect to a temporal variant of first-order logic.
The resulting logic expresses a wide range of temporal graph problems including a version of temporal cliques, an important notion when querying time series databases for community structures.
\end{abstract}

\section{Introduction}
\label{section.introduction}
The concept of algorithmic meta-theorems (AMT)~\cite{G08,K11} has arisen as a broad solution framework for entire classes of (static) graph problems.
A typical statement for AMT is: any graph problem expressible in X-logic can be solved in polynomial time on graphs of bounded Y-parameter, where X and Y are to be specified.
A weaker statement for AMT is: any graph problem expressible in X-logic can be solved in polynomial time on graphs of bounded Y-parameter under the condition that an appropriate decomposition is given as part of the input.
Three cornerstones have been settled for: weak AMT with First Order logic $\FO$ (playing~X) and twin-width (playing Y)~\cite{BKTW22}; typical AMT with Monadic Second Order logic $\MSO_1$ and clique-width~\cite{CMR00}; typical AMT with $\MSO_2$ and tree-width~\cite{CER93}.
Note that more general AMTs using $\FO$ logic exist, such as with nowhere dense graphs~\cite{GKS17} or those with bounded local clique-width~\cite{BDGKMST22}.
We focus on the three first results whose proper definitions are given in Section~\ref{section.definitions}, or in Refs.~\cite{BKTW22,CE12}.
The three results are incomparable as $\FO$ is less expressive than $\MSO_1$ which is less expressive than $\MSO_2$, whereas graphs of bounded tree-width are also of bounded clique-width which in turn are of bounded twin-width.

Although these results are not new, a number of AMT have been developed in the past few years.
For instance, $\FO$ can be enriched with an atomic predicate called \texttt{conn}, expressing the existence of a path joining two vertices while avoiding a given set of vertices. In turn, $\FO$+\texttt{conn} is less expressive than $\FO$+\texttt{DP}, where \texttt{DP} is a predicate expressing the existence of internally vertex-disjoint paths between a set of source vertices. These lead to results tying $\FO$+\texttt{conn} (again, playing X) with a parameter (playing Y) called the Haj\'os number~\cite{PSSTV22}, and a result tying $\FO$+\texttt{DP} with a parameter called the Hadwiger number in Ref.~\cite{GST23}.
Many other AMT for static graphs have been discovered too, from both centralised~\cite{BDJ23,BFLPST16} and distributed~\cite{CKM25,FMRT24} contexts.

A temporal graph is a sequence of static graphs over a discrete time span, called lifetime.
Unfortunately, far fewer positive results have been found for AMT on temporal graphs.
To begin with, it is notoriously a tricky question to find polynomial algorithms solving natural generalisations of graph problems on temporal graphs.
Most relevant polynomial cases are about paths, walks, backup-routes~\cite{BFJ03,VBP22,TAPBP25}, and path-based measures such as closeness~\cite{CMM20} and betweenness centrality~\cite{N25}.
Contrasting these few cases, the temporal version of spanning tree is widely known to be $NP$-complete~\cite{BF03}, leading to important current works on spanners, even when the (temporal) graphs are reduced to be trees themselves~\cite{CCS24}.
Temporal cliques and their generalisations to $\Delta$-cliques and $(\Delta,\gamma)$-cliques are $NP$-complete~\cite{VLM16,HMNS17,BP19}.
The duo of max-flow and min-cut translates into several $NP$-complete temporal flavours~\cite{ZFMN20}.
Even classically polynomial matching translates to various $NP$-complete problems~\cite{BB18,BBR20,MMNZZ20,MMNZZ23} which remain $NP$-complete on extremely restricted temporal graphs, notably with very small tree-width~\cite{MMNZZ20}.
Since matching can naturally be expressed in $\MSO_1$, this is intuitively a major obstruction for obtaining an AMT for temporal graphs based on tree-width:
any AMT tying a tree-width based parameter with a temporal version of $\MSO_1$ would have to either exclude the $NP$-complete problems introduced in~\cite{BB18,BBR20,MMNZZ20,MMNZZ23} or have the value of the parameter unbounded on the restricted temporal graphs used in the hardness results in~\cite{MMNZZ20}.

Our work strives to rectify the very lack of AMT for temporal graphs.
A temporal graph is a sequence of static graphs $\Gset=(G_t)_{t\in\llbracket0,\tau-1\rrbracket}$ over the same set of $n$ vertices.
Every graph $G_t$ in the sequence is called a snapshot graph.
Parameter $\tau$ is called the lifetime of $\Gset$.
The union graph $\Gset_\downarrow$ is the graph whose edge set contains all edges appearing at any snapshot in the sequence.
The static expansion graph $\Gset_\rightarrow$ is the graph where, in addition to the snapshot edges, $G_t$ and $G_{t+1}$ are joined by edges between the occurrences at $t$ and at $t+1$ of every vertex.
Three popular ways to lift a parameter $k(G)$ of a static graph $G$ to temporal graphs are, respectively, when defining $k_\infty$ as the maximum $k(G_t)$ over every snapshot $G_t$; when defining $k_\downarrow$ as the parameter $k(\Gset_\downarrow)$ of its union graph; and when defining $k_\rightarrow$ as the parameter $k(\Gset_\rightarrow)$ of its static expansion graph.
For parameter tree-width $\tw$, Ref.~\cite{FMNRZ20} sketched a high-level proof that no AMT can be found tying a temporal version of monadic second-order logic ($\MSOT$) with $\tw_\infty$, unless $P=NP$.
The cases of $\tw_\rightarrow$ and $\tw_\downarrow$ are left open, which we address in the sequel.
On the positive side, an AMT exists for $\MSOT$ when also the lifetime $\tau$ is bounded in addition to bounded $\tw_\downarrow$~\cite{ZFMN20}.
Up to our knowledge, no attempt in finding an AMT for temporal graphs has been successful without bounding the lifetime~\cite{ZFMN20,DELS26}, or the maximum number of time-edges incident with any vertex~\cite{EMMZ21}.
The latter parameter is $\Theta(n\times\tau)$.

We describe now our contributions, along with the organisation of the paper.
We formally define in Section~\ref{section.definitions} Monadic Second Order logic for temporal graphs, $\MSOtwo$ along with a restriction called $\MSOone$, and some temporal graph width parameters.
We show in Section~\ref{section.first_amt}
the first AMT for temporal graphs where the lifetime is not a parameter, as in the sense of fixed-parameter tractable (FPT) complexity~\cite{DF99,FG06,N06,CFKLMPPS15}:
Every problem expressible by a sentence in $\MSOtwo$ can be solved on input $\Gset$ in FPT time parameterized by the size of the sentence and the expanded tree-width $\tw_\rightarrow(\Gset)=\tw(\Gset_\rightarrow)$.
While the static expansion graph $\Gset_\rightarrow$ has a large number of vertices, $n\times\tau$, it is a globally sparse graph where the maximum clique has size at most $n$.
Besides, $0\leq\tw_\rightarrow(\Gset)\leq 2n$.
When for instance $\Gset$ has at most one edge per snapshot and $\tw_\rightarrow(\Gset)\leq2$, our AMT results in a polynomial time algorithm while any result with parameter $\tau$ similar to Ref.~\cite{ZFMN20} would have time complexity exponential in $\tau$ unless ETH fails (Exponential Time Hypothesis: $3$-SAT cannot be solved in subexponential time).
Even though $\tw_\rightarrow(\Gset)=O(n)$, to control its explosive growth in the large graph $\Gset_\rightarrow$, which easily forms grid minors from the snapshots' path minors, we below introduce smaller slices called differentials.
These differentials allow for bounding grid-like structures while preserving essential information about the graph's structural evolution over time.

We show some hardness results in the remainder of Section~\ref{section.first_amt}.
First, we show that finding an AMT that relates either $\tw_\infty$ or $\tw_\downarrow$ to $\MSOone$ is impossible, and therefore also to $\MSOtwo$, unless $P = NP$.
Moreover, for $\tw_\infty$ we obtain a hardness result even for a less expressive version of $\MSOone$, generalising the sketched high-level proof in Ref.~\cite{FMNRZ20}.
We then address a similar formalism for temporal twin-width, that are $\twww_\infty$, $\twww_\downarrow$ and $\twww_\rightarrow$.
We show the following result which can also be seen as the consequence of a recent result~\cite{Bonnet_arxiv} obtained independently and in parallel with our work.
We show that no AMT (even in its weak version) can be found tying neither $\twww_\infty$, $\twww_\downarrow$ nor $\twww_\rightarrow$ with $\MSOone$, unless $P=NP$.
Note that no AMT in the typical version can be found for twin-width since deciding whether the twin-width of a graph is $4$ is already an $NP$-complete problem~\cite{BBD22}.
Note also that the same result where $\MSOone$ is replaced by $\MSOtwo$ follows from the static twin-width hardness, since \textsc{HamiltonPath} is a $\MSO_2$ problem which remains $NP$-complete on graphs of twin-width $4$~\cite{BKTW22,IPS82}.
The above-mentioned recent result proving the $NP$-completeness of \textsc{Coloring} on graphs of twin-width $4$~\cite{Bonnet_arxiv} can also be used to deduce our result.
Essentially, these results hint about a certain tightness of our above AMT result tying $\tw_\rightarrow$ with $\MSOtwo$.

The rest of our results are described in Section~\ref{section.dtw}. Turning our attention back on the positive side, we relax $\MSOT$ down to First Order logic and define a local version called $\FOone^\Delta$.
This is a less expressive logic than First Order logic.
Nonetheless, we exemplify its expressivity by writing formulas in $\FOone^\Delta$ for problems related to temporal cliques, which are $NP$-complete and have received a lot of recent research interests~\cite{VLM16,HMNS17,BP19}.
We then show an AMT (weak version) tying $\twww_\rightarrow$ with $\FOone^\Delta$, that is, every problem expressible by a sentence in $\FOone^\Delta$ can be solved on input $\Gset$ and an appropriate decomposition for twin-width in FPT time parameterized by the size of the sentence and the expanded twin-width $\twww_\rightarrow(\Gset)=\twww(\Gset_\rightarrow)$.

We further investigate a finer definition of width parameter for temporal graphs, with tree-width and twin-width of smaller graphs than the static expansion graph $\Gset_\rightarrow$.
We introduce the derivative of $\Gset$ as subgraphs of $\Gset_\rightarrow$ induced by consecutive snapshots within a sliding time window.
Formally, the differential of $\Gset$ at $t$ over $\Delta$ is a static graph, defined as the static expansion graph of the $\Delta$ snapshots of $\Gset$ starting at $t$ and ending at $t+\Delta-1$.
We then show how every problem expressible by a sentence in $\FOone^\Delta$ can be solved on input $\Gset$ and an appropriate decomposition for twin-width in FPT time parameterized by the size of the sentence and the $\Delta$-twin-width, defined as the maximum twin-width over all differentials of $\Gset$.
As a by-product we also obtain the same result for $\Delta$-tree-width (albeit we do not need the assumption about the decomposition in the input, using \textit{e.g.}~\cite{RS95,B96,OS06}).

\section{Temporal Graphs Model Checking}
\label{section.definitions}
For a proper introduction to graph theory, refer to any standard textbook, \textit{e.g.}~\cite{BM08,D25,KT06}.
Unless stated otherwise, graphs in this paper are loopless simple undirected graphs.
For any graph $G$ we note $V(G)$ and $E(G)$ its vertex and edge sets.
If an edge exists between two distinct vertices $u\neq v$, we sometimes abusively note $uv=vu=\{u,v\}\in E(G)$.
A \textit{temporal graph} is a tuple $\Gset = (V, E_0, E_1, \ldots, E_{\tau-1})$, where $\tau\in\mathbb N$ is called the \emph{lifetime}, and every $G_t = (V, E_t)$ for $0\leq t<\tau$ is a graph called the \emph{snapshot at $t$}.
Equivalently, we also note $\Gset=(G_t)_{0\leq t<\tau}$ when the accent is put on the snapshots. 
The sizes of $G$ and $\Gset$ are $||G||=|V(G)|+|E(G)|$ and $||\Gset||=|V|+\Sigma_{t=0}^{\tau-1}|E_t|$, respectively.

The \emph{subgraph of $G$ induced by $V'\subseteq V(G)$} is noted $G[V']=(V',E')$ where $E'=\{uv\in E(G):u\in V'\land v\in V'\}$.
A \emph{path} of length $p\geq$ in $G$ is a sequence of $p$ successively adjacent edges joining $p+1$ distinct vertices.
A graph is \emph{connected} if there is a path joining any two distinct vertices.
A \emph{tree} is a connected graph over $n$ vertices and $n-1$ edges, for some $n\in\mathbb N$.

\paragraph*{Monadic Second Order Logic for Temporal Graphs}
A \emph{signature} $\sigma$ is a finite set of relation
symbols $R$, each with an arity $\ar(R)$.
A \emph{$\sigma$-structure} $\Aset$ consists of a non-empty set $D$, called \emph{universe} (or \emph{domain}), and,
for each $R\in\sigma$, a relation $R^\Aset\subseteq D^{\ar(R)}$.
For $D' \subseteq D$, the \emph{substructure of $\Aset$ induced by $D'$}, denoted $\struct{D'}^{\Aset}$, is the $\sigma$-structure with domain $D'$ and relations
$R^{\struct{D'}^{\Aset}}=R^\Aset\cap {D'}^{\ar(R)}$, for every symbol $R\in \sigma$.
The \emph{Gaifman graph} of $\Aset$ is the graph $\Gaif(\Aset)=(D, A_\Aset)$,where
$(a,b)\in A_\Aset$ if and only if $a\neq b$ and $a$ and $b$ appear together in some tuple of a relation $R^\Aset$.
Intuitively, it is the graph where two elements are linked by an edge if they are linked by some relation in the structure.

Model-Checking asks whether a structure satisfies a given property. In this
work, structures are (temporal) graphs. Monadic Second Order Logic has proven to be powerful enough to express many properties on graphs while staying decidable. 
Given two disjoint countable sets of symbols of variables,
$\mathcal{X}_0$ and $\mathcal{X}_1$ (first order and second order variables), an \MSO~ formula over a signature $\sigma$ (noted $\MSO(\sigma)$) is defined by
\begin{displaymath}
\varphi:= x=x'~|~ x\in X ~|~R(x_1,\dots, x_k) ~|~ \neg\varphi ~|~\varphi\land \varphi ~|~\exists x.\varphi ~|~\exists X.\varphi
\end{displaymath}
for $x,x',x_1,\dots, x_k\in\mathcal{X}_0$, $X\in \mathcal{X}_1$ and for all $R\in\sigma$ such that $\ar(R)=k$.

The semantics of an $\MSO(\sigma)$~formula over a structure $\Aset$ of domain $D$ is defined using two \emph{valuation functions} $\nu_0:\mathcal{X}_0\rightarrow D$ and $\nu_1:\mathcal{X}_1\rightarrow 2^{D}$. Given a valuation $\nu$, a variable $\alpha \in \mathcal{X}_0$ (\resp $\mathcal{X}_1$) and $m \in D$ (\resp $2^{D}$), we define the new valuation $\nu[\alpha\leftarrow m]$ by 
$\nu[\alpha\leftarrow m](\alpha)=m$, and $\nu[\alpha\leftarrow m](\beta)=\nu(\beta)$ for all $\beta$ in the definition domain of $\nu$.
 Then, a structure $\Aset$ over a domain $D$ is said to satisfy a formula $\varphi$ in $\MSO(\sigma)$, written $\Aset\models_{\nu_0, \nu_1}\varphi$ according to the following inductive definition:
 \begin{align*}
    \Aset&\models_{\nu_0,\nu_1} x=x' & \textrm{ if and only if } \nu_0(x)=\nu_0(x')\\
    \Aset&\models_{\nu_0,\nu_1} x\in X & \textrm{ if and only if } \nu_0(x)\in\nu_1(X)\\
    \Aset&\models_{\nu_0,\nu_1} R(x_1,\dots, x_k) & \textrm{ if and only if } (\nu_0(x_1),\dots,\nu_0(x_k))\in R^\Aset
     \textrm{ for all relation $R\in\sigma$}\\
    \Aset&\models_{\nu_0,\nu_1} \exists x.~\varphi & \textrm{ if and only if there exists $v\in D$ such that } \Aset\models_{\nu_0[x\leftarrow v],\nu_1}\varphi\\
    \Aset&\models_{\nu_0,\nu_1} \exists X.~\varphi & \textrm{ if and only if there exists $U\subseteq D$ such that }\Aset\models_{\nu_0,\nu_1[X\leftarrow U]}\varphi.
\end{align*}
The Boolean connectives $\neg\varphi$ and $\varphi\wedge \varphi'$ are interpreted by their usual meaning. We use the classical shortcuts $\varphi_1\vee \varphi_2:=\neg(\neg\varphi_1\wedge \neg \varphi_2)$, $\varphi_1\rightarrow \varphi_2:=
\neg \varphi_1\vee\varphi_2$ and $\forall \chi.~\varphi:= \neg(\exists \chi.~\neg \varphi)$, for $\chi\in\mathcal{X}_0\cup\mathcal{X}_1$.
Also, we will use the following shortcut
 $\exists! x.~ \varphi(x)$, for any formula $\varphi(x)$ to express the fact that there is a unique $x$ that satisfies $\varphi$.

Monadic Second Order logic over graphs is usually presented in two variants: $\MSO_1$, which quantifies
only the vertices, and $\MSO_2$, which quantifies over both vertices and edges.
$\MSO_2$ is strictly more expressive than $\MSO_1$.
We let $\MSOtwo$ be the logic $\MSO(\sigma_2)$ with $\sigma_2=\{\inc, \suc,\app\}$.
 A temporal graph $\Gset=(V,E_0,\dots, E_{\tau-1})$ is associated to a $\sigma_2$-structure $\structwo{\Gset}$ with domain 
$D_\Gset=V\cup (V\times \llbracket0, \tau-1 \rrbracket)\cup \bigcup_{0\leq t<\tau} (E_t\times \{t\})$. In particular, variables of the formula can designate either vertices or edges at different snapshots.
We use avatars of the vertices, that represent them in the different snapshots. 
We refer to vertices from $V$ by
\emph{untime vertices}, and to their avatars from $V\times \llbracket0, \tau-1 \rrbracket$ by \emph{time vertices}.
The symbol $\suc$ is interpreted by $\suc^{\structwo{\Gset}}=\{((v,t),(v,t+1))\mid v\in V\textrm{ and }0\leq t<\tau-1\}$, and relates two successive time vertices, and $\app^{\structwo{\Gset}}=\{(v,(v,t))\mid v\in V, 0\leq t\leq \tau-1\}$ relates an untime vertex to all of its time representatives. 
The interpretation of the relation $\inc$ associates a temporal edge with a time vertex to which it is incident, i.e., $\inc^{\structwo{\Gset}}=\{((vv',t), (v,t))\mid 0\leq t<\tau_1\textrm{ and } vv'\in E_t\}$.

$\MSOone$ uses the signature $\sigma_1=\{\adj,\suc,\app\}$.
We interpret a temporal graph $\Gset=(V,E_0,\dots, E_{\tau-1})$ by a $\sigma_1$-structure
$\structone{\Gset}$ over the domain $D_\Gset^1=V\cup (V\times\llbracket 0, \tau-1\rrbracket)$.
The relation $\suc$ is interpreted by $\suc^{\structone{\Gset}}=\{((v,t),(v,t+1))\mid v\in V\textrm{ and }0\leq t<\tau-1\}$, and relates two successive time vertices, and $\app^{\structone{\Gset}}=\{(v,(v,t))\mid v\in V, 0\leq t\leq \tau-1\}$ relates an untime vertex to any of its time representatives. 
The relation $\adj$ is interpreted by $\adj^{\structone{\Gset}}=\{((u,t),(v,t))\mid u,v\in V, 0\leq t\leq \tau-1\textrm{ and } uv\in E_t\}$, simply relating two time vertices when they are adjacent in the snapshot in which they stand; it is hence a symmetric relation. 
Elements of a $\sigma_1$-structure are shown in Figure~\ref{fig:rep_expansion_time_graph}.
\begin{figure}
	\centering
\scalebox{.8}{
\begin{minipage}{.6\textwidth}
	\centering
    \def\hsep{3.5}
    \def\vsep{1.5}
	\begin{tikzpicture}[scale=.5,dot/.style={draw,circle,minimum size=0.15cm,inner sep=0pt,outer sep=0pt,fill=black}]
	\path \foreach \i in {0,...,4}
			\foreach \j in {0,...,3}
			{coordinate [dot] (u\i\j) at (\hsep*\j,-\vsep*\i)};

    \draw[dashed, line width=0.2ex] \foreach \i in {0,...,4}
	\foreach \j in {0,...,2}
	{(u\i\j)--(u\i\the\numexpr\j+1)};

	\draw[line width=0.3ex] (u00) to[bend right=45] (u20)
	(u01)--(u11)
	(u10) to[bend left=45] (u30)
	(u22) to[bend right=45] (u42)
	(u21)--(u31)
	(u22)--(u32);

	\foreach \i in {0,...,4}
		\node[draw = purple!80] (u\i)at(-\hsep,-\vsep*\i) {$u_\i$};
	\foreach \j in {1,...,4}
		\node (t\j)at(\hsep*\j - \hsep, -5*\vsep) {$t = \the\numexpr\j-1$};


    \draw[line width=0.2ex, purple!80, ->] (u2) to [] node[below] {$\app$} (u20);
    \draw[line width=0.2ex, purple!80, ->] (u2) to [bend right=30] node[right=0.7cm] {$\app$} (u21);
    \draw[line width=0.2ex, purple!80, ->] (u2) to [bend left=20] node[below] {$\app$} (u22);
    \draw[line width=0.2ex, purple!80, ->] (u12) to [bend left=20] node[above] {$\suc$} (u13);
    \draw[line width=0.2ex, purple!80, ->] (u01) to[bend left=45] node[right] {$\adj$} (u11);
    \draw[line width=0.2ex, purple!80, ->] (u11) to[bend left=45] node[left] {$\adj$} (u01);

	\end{tikzpicture}
\end{minipage}
\begin{minipage}{.4\textwidth}
	\centering
      	\begin{tikzpicture}[scale =0.35, state/.style={draw=black!50,circle,very thick,fill=orange!30}]
		\node[state] (u0) at (0,-1.5) {$u_0$};
		\node[state] (u1) at (3,0) {$u_1$};
		\node[state] (u2) at (3,-3) {$u_2$};
		\node[state] (u3) at (6,-1.5) {$u_3$};
		\node[state] (u4) at (3,-6) {$u_4$};

		\path[black, line width=0.2ex]
		(u0) edge (u1)
		(u1) edge (u3)
		(u0) edge (u2)
		(u2) edge (u3)
		(u2) edge (u4);
	\end{tikzpicture}
\end{minipage}}
\caption{(left, black edges) A temporal graph $\Gset$; (left, purple edges) some relations of its $\sigma_1$-structure $\structone{\Gset}$,
its static expansion graph is pictured when taking both black and dashed edges;
(right) the union graph of $\Gset$.}
\label{fig:rep_expansion_time_graph}
\end{figure}
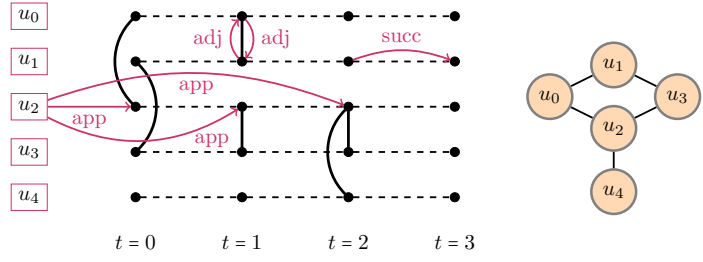

The \emph{free variables} of a formula $\varphi$ are the variables $x\in\mathcal{X}_0\cup\mathcal{X}_1$ that are not under the scope of a quantifier $\exists$. 
When $x_1,\dots, x_n$ are the free variables of $\varphi$, we may make them explicit by writing $\varphi(x_1,\dots, x_n)$.
When $\varphi$ has no free variable, it is called a \emph{sentence}. In that case, we do not need valuation functions to interpret the formula and simply write $\Aset\models\varphi$.
The quantifier rank of a formula $\varphi$, noted $\qr(\varphi)$ is the maximum number of nested quantifiers appearing in the formula. The size of a formula $\varphi$, denoted by $|\varphi|$,
is the total number of symbols used to build it : $|x=x'|=|x\in X|=|R(x,x')|=1$, for all $R$ in the signature, and $|\neg \varphi|=|\exists x.~\varphi|=|\exists X.~\varphi|=|\varphi|+1$, and 
$|\varphi\wedge\varphi'|=|\varphi|+|\varphi'|+1$.
In Section~\ref{subsec:mso-ex}, we give sentences to express two fundamental extensions of classical graph matching and $1$-connectivity.
As recalled in the introduction, temporal matching has many definitions which lead to $NP$-complete problems.
At the same time, deciding the complexity class of the various temporal versions of graph cuts is a quite involved task.
To give an idea of these problems, the following definition is one specific definition of temporal matching, \resp temporal cut, that is used in the literature.
\begin{itemize}
\item \textbf{Temporal matching:} does there exist a set of $k$ pairwise independent $\gamma$-edges in a given temporal graph? A $\gamma$-edge is a set of consecutive time edges over $\gamma$ snapshots between two same vertices. Two $\gamma$-edges are independent if they do not share any time vertex.
\item \textbf{Temporal cut of size $1$:} does there exist a vertex whose removal makes an input temporal graph not a temporally connected graph? A temporal graph is temporally connected if every pair of source and target vertices are starting and ending points of a temporal path.
\end{itemize}

\paragraph*{FPT Complexity and Width Parameters}
For a proper introduction to FPT complexity, refer to any standard textbook, \textit{e.g.}~\cite{DF99,FG06,N06,CFKLMPPS15}.
A parameterized problem is a language $L \subseteq \Sigma^* \times \mathbb N$, where $\Sigma$ is a finite alphabet.
It is FPT parameterized by $k$ if deciding whether $(x,k) \in L$ can be done in time $f(k) \cdot |x|^{\mathcal{O}(1)}$, where $f$ is an arbitrary function depending only on $k$.

	Let $G$ be a (static) graph.
	Let $T=(V_T,E_T)$ be a tree.
        For each $u\in V_T$, let $B_u\subseteq V(G)$ be a subset called a \emph{bag}.
	The tuple $\Tbb = (T, \{B_u:u \in V_T\})$
	is a \emph{tree decomposition} of $G$ if it respects the following
	three conditions:
	        $(i)$ $V(G)=\cup_{u \in V_T} B_u$;
	        $(ii)$ for every edge $vw \in E(G)$, there exists a node $u\in V_T$
		such that $\{v, w\} \subseteq B_u$;
                $(iii)$ and for every $v \in V(G)$,
		the subgraph of $T$ induced by $\{u\in V_T: v\in B_u\}$ is a tree.
        The \emph{width} of $\Tbb$ is
	$\width(\Tbb)=\max\{|B_u| - 1:u \in V_T\}$.
	The \emph{tree-width} of $G$ is the minimum width over all
	tree decompositions of $G$, that is,
	$\tw(G) = \min\{\width(\Tbb):\Tbb\textrm{ is a tree decomposition of } G \}$.
        If $G$ is a $\sigma$-structure the tree-width of $G$ is defined as the tree-width of its Gaifman graph.
Numerous important results have stemmed from Courcelle's Theorem, which is:
\begin{theorem}[\cite{Courcelle90:graphrewriting}]\label{th:courcelle}
For every integer $w\in\mathbb N$ and every sentence $\varphi$ of \MSO{} over a signature $\sigma$ there
is a linear time algorithm that decides whether a given $\sigma$-structure $\Aset$ of tree-width at most $w$ satisfies $\varphi$.
That is, the complexity of \MSO{} model checking is linear-time FPT parameterized by $w$ and $|\varphi|$.
\end{theorem}

We now address sequences of graphs $G_i=(V_i,E_i)$, for $n\geq i>0$, with $E_i=B_i\uplus R_i$ a partition into two sets of the edge set into edges called black and red edges, respectively.
The red-degree of a vertex is the number of red edges incident to this vertex.
The red-degree of $G_i$ is the maximum red-degree of its vertices.
Such a sequence $(G_i)_{n\geq i>0}$ is called a $d$-\textit{contraction} of $G=(V,E)$ if $V_n=V$, $B_n=E$, $|V_1|=1$, and for $n>i>0$ we have that $G_i$ has red-degree at most $d$ and also that $G_i$ is obtained from $G_{i+1}$ by the following process.
Replace two and only two vertices $u,v$ of $G_{i+1}$ by a new vertex $w$ in $G_i$ where:
$wx \in B_{i}$ if $\{ux, vx\} \subseteq B_{i+1}$;
$xy\in B_i$ if $xy\in B_{i+1}$ and $\{x,y\}\cap\{u,v\}=\emptyset$;
otherwise, $xy\in R_i$ if (it does not belong to the previous two cases and if) $xy\in B_{i+1}\cup R_{i+1}$.
The \textit{twin-width of $G$}, $\tww{G}$, is the minimum value of $d$ taken over all possible $d$-contractions of $G$.

Let $\Gset=(V,E_0,\dots, E_{\tau-1})$ be a temporal graph.
The \emph{union graph of $\Gset$}, $\unionG{\Gset}$, is the graph containing all edges of the snapshots, that is,
$\unionG{\Gset} = (V,\cup_{0\leq t<\tau}E_t)$.
The \emph{static expansion graph of $\Gset$}, $\expG{\Gset}$, has as vertices all avatars of $V$ appearing in the snapshots, and as edges two types that we sometimes refer to as \emph{temporal edges} and \emph{local edges}, \textit{cf.}~Figure~\ref{fig:rep_expansion_time_graph}.
Formally, $\expG{\Gset}=(V\times\llbracket 0,\tau-1\rrbracket,\tempE\cup \adjE)$ where
$\adjE=\{(u,t)(v,t)\mid 0\leq t<\tau,  u,v\in V\textrm{ and }uv\in E_t\}$
and
$\tempE=\{(u,t)(u,t+1)\mid0\leq t<\tau-1,u\in V\}$.
Three temporal extensions of tree-width and twin-width follow, namely:
\begin{itemize}
\item \textbf{Infinite tree-width and twin-width:}
$\tw_{\infty}(\Gset)$ and $\twww_{\infty}(\Gset)$ which are the maximum widths over
any snapshot $G_t$, that is, $\tw_{\infty}(\Gset)=\max\{\tw(G_t):0\leq t<\tau\}$ and $\twww_{\infty}(\Gset)=\max\{\twww(G_t):0\leq t<\tau\}$,
\item \textbf{Underlying tree-width and twin-width:}
$\tw_\downarrow(\Gset)=\tw(\Gset_{\downarrow})$ and
$\twww_\downarrow(\Gset)=\twww(\Gset_{\downarrow})$
which are the widths of its union graph,
\item \textbf{Expanded tree-width and twin-width:}
$\ttw(\Gset)= \tw(\Gset_{\rightarrow})$ and
$\ttwww(\Gset)= \twww(\Gset_{\rightarrow})$
which are the widths of its static expansion graph.
\end{itemize}

\section{$\MSOT$ Model Checking}
\label{section.first_amt}
The following theorem's statement holds for $\MSOtwo$, hence for $\MSOone$ too.
We first make useful remarks for some shortcuts in its proof.

 We will make use of the classical reflexive and transitive closure of a relation, especially for the relation $\suc$: let $\closedsuc(X):=\forall u\forall v.~ (u\in X\wedge\suc(u,v))\rightarrow v\in X$. This formula expresses the fact that the set $X$ is closed by the relation $\suc$. Now we let 
 $\sucstar(x,x'):=\forall X.~(x\in X\wedge\closedsuc(X))\rightarrow x'\in X$ that is true whenever $x$ corresponds to the vertex 
 $(v,t)$ and $x'$ to $(v,t')$ with $t\leq t'$. For $i\geq 2$, we will also make use of the formula $\suc_i(x,x'):=\exists x_1\exists x_2\dots \exists x_{i-1}.~ (\suc(x,x_1)\wedge \bigwedge_{k=1}^{i-2} \suc(x_k,x_{k+1})\wedge
 \suc(x_{i-1}, x'))$ to express the fact that $x'$ represents the same untime vertex than $x$, but in a snapshot appearing $i$ instants later, and $\suc_{\leq i}(x,x'):=(x=x')\vee \suc(x,x')\vee \bigvee_{k=2}^i
 \suc_k(x,x')$. 
 We can also use $\timevertex(x):=\exists x'.~\app(x',x)$ to express the fact that the element $x$ is a time vertex of the structure.

Observe that all these formulas are also in $\MSOone$. 
\begin{theorem} \label{th:mso1ttw}
	Given an $\MSOtwo$ sentence $\varphi$ and a temporal graph $\Gset$, we can decide if
	$\structone{\Gset} \models \varphi$ in 
	$O(f(\tw_\rightarrow(\Gset),|\varphi|)\cdot ||\Gset||)$, with $f$ a computable function only depending on $|\varphi| \text{ and } \allowbreak \tw_\rightarrow(\Gset)$.
\end{theorem}
\begin{proof}\footnote{An alternative proof would be as follows.
The tree-width of $\structone{\Gset}$ is at most twice the expanded tree-width of $\Gset$:
given a tree decomposition for $\Gset_\rightarrow$, add in every bag of this tree decomposition all the untime vertices connected to some timed vertex already contained in the bag.
The resulting tree is a tree decomposition for $\structone{\Gset}$:
the set of untime vertices is an independent set;
the set of all timed vertices of a given untime vertex is a connected set, hence their bags induce a connected subtree.
Then Courcelle's Theorem (Theorem~\ref{th:courcelle}) allows us to conclude for both $\MSOone$ and $\MSOtwo$.}
We first argue for $\MSOone$ then for $\MSOtwo$.

We turn $\varphi$ into an $\MSO$ sentence over the signature $\sigma'=\{\suc, \adj\}$. A temporal graph $\Gset$ is turned into a $\sigma'$-structure $\struct{\Gset}$ over the domain $D_\Gset=V\cup (V\times \llbracket0, \tau-1 \rrbracket)$. We interpret the relations as follows: $\adj^{\struct{\Gset}}=\adj^{\structwo{\Gset}}=\{((v,t), (v',t))\mid 0\leq t<\tau\textrm{ and } vv'\in E_t\}$, and $\suc^{\struct{\Gset}}=\suc^{\structwo{\Gset}}\cup \{(v,(v,0))\mid v\in V\}$. While the relation $\adj$ is 
interpreted as before, the relation $\suc$ is now interpreted as relating two successive time vertices, as in the static expansion graph, \emph{and} relating each \emph{untime} vertex $v$ only to $(v,0)$, its representative in the first snapshot.  We explain how to transform $\varphi$ into a formula $\varphi'\in \MSO$ over the signature
$\sigma'$ such that, for any temporal graph $\Gset$, $\structone{\Gset}\models\varphi$ if and only if $\struct{\Gset}\models\varphi'$. This is done inductively.

The only case that needs to be transformed
is when $\varphi=\app(x,y)$. We do not have this relation anymore, but we can check that the time vertex $y$ is indeed a representative of the untime vertex $x$ by checking if $y$ is accessible from
$x$ via a formula $\suc^*(x,y)$, that expresses the reflexive and transitive closure of $\suc$, described in above remarks,
along with the formula $\timeformula(x) = \exists x'.~\suc(x',x)$ to express the fact that the element $x$ is a time vertex of the structure.

Hence, we can let $(\app(x,y))'=\neg\timeformula(x)\wedge\timeformula(y)\wedge\suc^*(x,y)$. The rest of the inductive transformation from $\varphi$
to $\varphi'$ is straightforward. Moreover, it is easy to see that $\ttw(\Gset)=\tw(\struct{\Gset})$. Courcelle's Theorem (Theorem~\ref{th:courcelle}) allows us to conclude.

The extension of the theorem to $\MSOtwo$ comes for free because Courcelle's theorem holds for $\MSO_2${} too. 
We now turn $\varphi$ into an $\MSO$ sentence over the signature $\sigma''=\{\suc, \inc\}$. We transform now a temporal graph $\Gset$ into a $\sigma''$-structure $\struct{\Gset}_2$ over the domain 
$D_\Gset=V\cup (V\times \llbracket0, \tau-1 \rrbracket)\cup \bigcup_{0\leq t<\tau} (E_t\times \{t\})$. We interpret the relations as follows: $\inc^{\struct{\Gset}}=\inc^{\structwo{\Gset}}=\{((vv',t), (v,t))\mid 0\leq t<\tau\textrm{ and } vv'\in E_t\}$, and $\suc^{\struct{\Gset}}=\{((v,t),(v,t+1))\mid v\in V\textrm{ and }0\leq t\leq \tau-1\}\cup \{(v,(v,0))\mid v\in V\}$. While the relation $\inc$ is 
interpreted as before, as the incidence relation between an edge and a time vertex, the relation $\suc$ is again interpreted as relating two successive time vertices, as in the static expansion graph, \emph{and} relating each \emph{untime} vertex $v$ only to $(v,0)$, its representative in the first snapshot.  
The transformation of $\varphi$ into $\varphi'\in \MSO(\sigma'')$ is done inductively and follows the same scheme as the proof of Theorem~\ref{th:mso1ttw}: again, the only case that needs to
be transformed is when $\varphi=\app(x,y)$, and we still let $(\app(x,y))'=\neg\timeformula(x)\wedge\timeformula(y)\wedge\suc^*(x,y)$. Again, classical calculations allow to see that 
$\ttw(\Gset)=\tw(\struct{\Gset}_2)$.
Courcelle's Theorem (Theorem~\ref{th:courcelle}) allows us to conclude.
\end{proof}
\begin{figure}
    \centering
\scalebox{.8}{
    \begin{minipage}{.5\textwidth}
        \def\hsep{2.5}
        \def\vsep{2}
        \def\bending{45}
	\begin{tikzpicture}[scale=.4,dot/.style={draw,circle,minimum size=0.15cm,inner sep=0pt,outer sep=0pt,fill=black}]
	\path \foreach \i in {0,...,3}
			\foreach \j in {0,...,3}
			{coordinate [dot] (u\i\j) at (\hsep*\j,-\vsep*\i)};

    \foreach \i in {0,...,3}
		\node[] (u\i)at(-2,-\vsep*\i) {$u_\i$};
    
    \draw[line width=0.2ex] \foreach \i in {0,...,3}
	\foreach \j in {0,...,2}
	{(u\i\j)--(u\i\the\numexpr\j+1)};

	\draw[line width=0.3ex] (u00) to[bend right=45] (u20)
	(u01)--(u11)
	(u10) to[bend left=45] (u30)
	(u03) to[bend left=45] (u33)
	(u21)--(u31)
	(u22)--(u32);

	\end{tikzpicture}
    \end{minipage}%
    \begin{minipage}{.4\textwidth}
        \def\hsep{2.5}
        \def\vsep{2}
        \def\bending{45}
        \def\posa{0.3}
        \def\posb{0.7}
    \begin{tikzpicture}[scale=.4,dot/.style={draw,circle,minimum size=0.15cm,inner sep=0pt,outer sep=0pt,fill=black}]
	\path \foreach \i in {0,...,3}
			\foreach \j in {0,...,3}
			{coordinate [dot] (u\i\j) at (\hsep*\j,-\vsep*\i)};

    \foreach \i in {0,...,3}
		\node[draw = purple!80] (u\i) at(-3,-\vsep*\i) {$u_\i$};
    
    \foreach \i in {0,...,3}{
	    \foreach \j in {0,...,2}
	        {\draw[dashed, ->, line width=0.2ex] (u\i\j) -- (u\i\the\numexpr\j+1);}}

    \foreach \i in {0,...,3}
	        {\draw[dashed, ->, line width=0.2ex] (u\i) -- (u\i0);}

    
    \draw[line width=0.2ex, blue, ->] (u20) to[bend left=45] (u00);
    \draw[line width=0.2ex, blue, ->] (u00) to[bend left=45] (u20);

    \draw[line width=0.2ex, blue, ->] (u11) to[bend left=45] (u01);
    \draw[line width=0.2ex, blue, ->] (u01) to[bend left=45] (u11);

    \draw[line width=0.2ex, blue, ->] (u30) to[bend left=45] (u10);
    \draw[line width=0.2ex, blue, ->] (u10) to[bend left=45] (u30);

    \draw[line width=0.2ex, blue, ->] (u31) to[bend left=45] (u21);
    \draw[line width=0.2ex, blue, ->] (u21) to[bend left=45] (u31);

    \draw[line width=0.2ex, blue, ->] (u32) to[bend left=45] (u22);
    \draw[line width=0.2ex, blue, ->] (u22) to[bend left=45] (u32);

    \draw[line width=0.2ex, blue, ->] (u03) to[bend left=30] (u33);
    \draw[line width=0.2ex, blue, ->] (u33) to[bend left=30] (u03);

	\end{tikzpicture}

    \end{minipage}}
    \caption{(left) A temporal static expanded graph $\Gset_\rightarrow$;
    (right) Its associated structure $\struct{\Gset}$;
    (right, blue edges) Its predicate $\adj$;
    (right, dashed edges) Its predicate $\suc$.}
    \label{fig:oriented-graph}
\end{figure}

\paragraph*{Obstructions to Meta-Theorems on $\MSOone$}
In the sequel, Theorem~\ref{absence_meta_extended}, Corollary~\ref{corollary_infinite_twinwidth}, Theorem~\ref{thm:fpt-underlying-twww} also follow from the recent result stating $NP$-completeness of \textsc{Coloring} on graphs of twin-width $4$~\cite{Bonnet_arxiv}.
For the sake of completeness, we also give an independent proof for each case.

We show now that the expanded tree-width is the only measure of the graphs among the ones we introduced that allows to exhibit such a meta-theorem on temporal graphs. 
Then, we show in the remainder of the section a number of other negative results.
We leave open whether any of these results implies the para-$NP$-hardness\footnote{An FPT problem $L \subseteq \Sigma^* \times \mathbb N$ parameterized by $k$ is said to be para-$NP$-hard if it is $NP$-complete to decide whether $(x,k)\in L$ for $k$ bounded by a constant.} of the corresponding problem.
\begin{theorem}\label{th:neg_tw_inf}
	Assuming $P \neq NP$, there is no FPT algorithm parameterized by the combined value (or the sum) of the infinite tree-width $\tw_{\infty}$, the lifetime of the input temporal graph, and the size of the formula for $\MSOone{}$ model checking, even if we restrain the
	formula to the class of \MSOone{} formulas written
	without the symbol $\app$, or to formulas written without the symbol $\suc$. 
\end{theorem}
\begin{proof}
	The proof is inspired by the arguments given by Fluschnik et al.~\cite{FMNRZ20}.
	Assume for the sake of contradiction that we have such an FPT algorithm.
	Consider an instance $G=(V,E)$ of the $3$-coloring problem such that
	the degree of $G$ is at most 4. Vizing's theorem~\cite{MJG92} ensures that
	in polynomial time w.r.t. $||G||$, there exists an application $c:E\rightarrow \{0,1,2,3,4\}$
	such that, if $e$ and $e'$ are two different edges incident to the same vertex, $c(e)\neq c(e')$.
	We build the temporal graph
	$g(G) = (V, E_0, E_1, E_2, E_3, E_4)$ such that for all $i \in \{0,1,2,3,4\}$,
	$uv \in E_i$ if and
	only if $c(uv)=i$.
	Thus,
	we have $g(G)_\downarrow = G$.
	We show in below paragraph how this coloring gives $\tw_{\infty}(g(G)) = 1$, and how to build
	an $\MSOone$ formula $\varphi_{3\text{col}\downarrow}$ which expresses
	the $3$-colorability of the union graph of a given temporal graph, without using the symbol $\app$. 
	Note that it is possible to write another formula that expresses $3$-colorability without the symbol $\suc$. 
	Then, for every graph $G$ having a degree bounded by $4$,
	we could decide $3$-colorability of $G$ in
	$O(f(\tw_\infty(g(G)),|\varphi_{3\text{col}\downarrow}|, 5)\cdot ||g(G)||^d)$,
	\ie $O(f(1,|\varphi_{3\text{col}\downarrow}|, 5)\cdot ||G||^d)$.
	This would result in a polynomial time algorithm to decide 3-colorability
	on the class of graphs of degrees bounded by $4$, which is known to
	be $NP$-complete~\cite{D80}.

Consider an instance $G=(V,E)$ of the $3$-coloring problem such that
	the degree of $G$ is at most 4. Vizing's theorem~\cite{MJG92} ensures that
	in polynomial time w.r.t. $||G||$, there exists an application $c:E\rightarrow \{0,1,2,3,4\}$
	such that, if $e$ and $e'$ are two different edges incident to the same vertex, $c(e)\neq c(e')$.
	We build the temporal graph
	$g(G) = (V, E_0, E_1, E_2, E_3, E_4)$ such that for all $i \in \{0,1,2,3,4\}$,
	$uv \in E_i$ if and
	only if $c(uv)=i$.
	Thus,
	we have $g(G)_\downarrow = G$.
	This coloring gives an efficient tree decomposition of each snapshot $g(G)_t$ of $g(G)$:
	for each edge $e=uv\in E_t$, we build a bag $B_e=\{u,v\}$ that trivially covers $e$.
	For each node $v\in V$ such that $v$ is not in any edge of $E_t$, we build a bag $B_v$. Then,
	by the property of the coloring application $c$, all the bags are pairwise disjoint, ensuring that any tree $T_t$
	over these bags is a tree decomposition of $g(G)_t$. It results in $\tw_{\infty}(g(G)) = 1$.

We explain how to build $\varphi_{3\text{color}\downarrow}$.
First, we define $\texttt{Partition}(X_1, X_2, X_3) = \forall x.~(x \in X_1 \wedge x \notin X_2 \wedge x \notin X_3) \vee (x \notin X_1 \wedge x \in X_2 \wedge x \notin X_3) \vee (x \notin X_1 \wedge x \notin X_2 \wedge x \in X_3)$ that is true if and only if the disjoint union of $X_1, X_2$ and $X_3$ is equal to the whole space. Also,
we define $\psi(X) = \exists s\exists s'.~ s \in X \wedge s' \in X \wedge \adj(s,s')$, which is true if and only if there is an edge between two time vertices of $X$. Then, we can define $\varphi_{3\text{col}\downarrow}= \exists X_1 \exists X_2\exists X_3.~ (\texttt{Partition}(X_1, X_2, X_3) \wedge \closedsuc(X_1) \wedge \closedsuc(X_2) \wedge \closedsuc(X_3) \wedge \neg \psi(X_1) \wedge \neg \psi(X_2) \wedge \neg \psi(X_3))$. \\
The fact that each part $X_i$ of the set of vertices and time vertices is closed by the relation $\suc$ implies that all the avatars of a given untime vertex are in
the same part $X_i$ (the untime vertex itself is not bound to be in the same set $X_i$ but it does not matter). 
It then stands that $\structone{\Gset}\models\varphi_{3\text{col}\downarrow}$ if and only if the set of vertices can be divided into three disjoint sets such that there exist no edge, in any snapshot, 
between vertices of the same set, which is exactly $\Gset_\downarrow$ being $3$-colorable.

Another way to write $\varphi_{3\text{color}\downarrow}$, without the symbol $\suc$ but using the relation $\app$ is the following:
	$\varphi_{3\text{col}\downarrow} =\exists X_1\exists X_2\exists X_3.~\texttt{Partition}(X_1, X_2, X_3)\wedge(\forall u\forall v\forall s\forall t.~
		 \app(u,s) \land \app(v,t) \land
					\adj(s,t) \rightarrow
							\displaystyle\bigwedge_{i = 1}^3 \neg 
								(u \in X_i \land v \in X_i))$

Here, we partition all the vertices, time and untime ones, but we are only concerned with how the \emph{untime vertices} are partitioned. The formula requires that
if two avatars of two untime vertices are linked in some snapshot, then the two untime vertices should not appear in the same set.
\end{proof}

\begin{theorem}\label{th:neg_rtb}
	Assuming $P \neq NP$, there is no FPT algorithm parameterized by the combined value (or the sum) of the underlying tree-width $\tw_\downarrow$ of the input temporal graph and the size of the input formula for \MSOone{} model checking.
\end{theorem}
\begin{proof}
The proof of the next theorem makes use of a problem called Return-To-Base Temporal Graph Exploration (RTB-TGE), which is NP-hard even if the input temporal graph $\Gset$
	is such that $\tw_\downarrow(\Gset)\leq 2$~\cite{AMSR21,BVZ19}. 
		We give in the below paragraph the precise definition of this problem and show how to define a formula $\varphi_{\text{RTB-TGE}}(u)$, in $\MSOone$, that is true if and only if the graph $\Gset$ with vertex $u$ is a positive instance of RTB-TGE. We also explain there that the existential version of this problem, where we ask, given a graph $\Gset$, whether there exists a vertex $u$ such that 
		$(\Gset,u)$ is a positive instance of RTB-TGE is also NP-hard, even if $\tw_\downarrow(\Gset)\leq 3$. %
	An FPT algorithm would decide if a temporal graph $\Gset$ such that $\tw_\downarrow(\Gset)\leq 3$ is a solution of existential RTB-TGE in $O(f(\tw_\downarrow(\Gset),|\exists u.~ \varphi_{\text{RTB-TGE}}(u)|)\cdot ||\Gset||^d)$, i.e. in $O(f(3,|\exists u.~\varphi_{\text{RTB-TGE}}(u)|)\cdot ||\Gset||^d)$, 
	contradicting NP-hardness of the problem. 

Given a temporal graph $\Gset = (V, E_0,\dots, E_{\tau-1})$ and two vertices $v,v'\in V$,
	a \emph{strict temporal walk} from $v$ to $v'$ is a sequence of temporal edges
	$v_1v_2, v_2v_3, \ldots, \allowbreak v_pv_{p+1}$, such that $v=v_1$, $v'=v_{p+1}$,
	and for each $1\leq i\leq p$, $v_iv_{i+1}\in E_{t_i}$ with
	$t_1 < t_2 < \ldots < t_p$.
	The \textsc{Return-To-Base Temporal Graph Exploration (RTB-TGE)}, asks, given a temporal graph $\Gset = (V, E_0,\dots, E_{\tau-1})$ and a designated vertex $u \in V$, if there is a strict
	temporal walk starting and ending at $u$ and visiting all the vertices of $V$.

We define the formula $\varphi_{\text{RTB-TGE}}$.
Note that the size of $\varphi_{\text{RTB-TGE}}$ is independent of the size of any temporal graph satisfying it.

We use the \MSOone{} formula $\timeadj(s,s')=\exists t.~(\suc^*(s,t)\wedge s\neq t \wedge \adj(t,s'))$, expressing the fact that $s$ and $s'$ are interpreted by two vertices $(u,t)$ and $(v,t')$ such that $t<t'$ and $uv\in E_{t'}$.
Recall that $\suc^*$ has been defined in the beginning of Section~\ref{section.first_amt}, and is also a formula of \MSOone.
We also use $\timevertex(x)=\exists x'.~\app(x',x)$ to express that $x$ is a time vertex. 
\begin{align*}
	\varphi_{\text{RTB-TGE}}(u) = &&\\
	\exists X.~ & [\forall x.~ x \in X \rightarrow \timevertex(x)] \land [\forall v.~(\timevertex(v)\vee(\exists x.~ x \in X \land \app(v, x)))] &\land \\
	& \exists s_\text{start} \exists s_\text{end}.~ [s_{\text{start}} \in X \land s_{\text{end}} \in X \land \app(u,s_{\text{start}}) \land \app(u,s_{\text{end}})] &\land \\
	& \forall s.~ [s \in X \land s \neq s_{\text{start}} \land s \neq s_{\text{end}} \rightarrow \exists!s_{\text{prev}} \exists!s_{\text{suc}}.~ &\\
					&s_{\text{prev}} \in X \land s_{\text{succ}} \in X \land \timeadj(s_{\text{prev}}, s) \land 
					\timeadj(s, s_{\text{suc}})] &\land \\
		& \exists s.~ [s \in X \land \timeadj(s_{\text{start}}, s)] &\land \\
		& \exists s.~ [s \in X \land \timeadj(s, s_f)]
\end{align*}

The two first conjuncts express the fact that, in a temporal graph $\Gset=(V,E_0,\dots, E_{\tau-1})$ satisfying the formula,
the set $X$ contains only time vertices, and that each vertex $v\in V$ is represented in $X$. The second conjunct identified
the time vertices that start and end the temporal walk in $\Gset$, while the last conjuncts express the fact that all the 
time vertices in $X$ form a strict walk in the graph. Since the two identified time vertices that start and end the temporal 
walk represent the vertex $u$, this enforces the existence of a strict temporal walk, starting and ending in $u$, and
visiting all the vertices of $V$.

Consider an RTB-TGE instance $(\Gset, u)$. We show that we can consider equivalently an instance $\Gset'$ of an existential version of the RTB-TGE problem.
	Introducing a new vertex $\bar{u}$, we construct the temporal graph $\Gset' = (V \cup \{\bar{u}\}, \{u\bar{u}\}, E_0, \ldots, E_{\tau-1}, \{u\bar{u}\})$. Observe that $\Gset'$ can be constructed in polynomial time in the size of $\Gset$. If there exists $v \in V \cup \{\bar{u}\}$ such that a strict temporal walk $v_1v_2, v_2v_3, \ldots, v_pv_{p+1}$ starting at $v$, ending at $v$ and visiting all the vertices of $V \cup \bar{u}$ exists in $\Gset'$, then $\bar{u} =v$, since the only edges incident to $\bar{u}$ occur at first and last instants. 
	The sequence $v_2v_3, v_3v_4, \ldots, v_{p-1}v_p$ is then a strict temporal walk in $\Gset$ starting at $u$, ending at $u$ and visiting all the vertices of $V$. Then, Existential-RTB-TGE is also
	NP-hard. 
	
	Moreover, any tree decomposition $\Tbb = (T, \{B_v: v \in T\})$ of $\Gset_\downarrow$ can be adapted into $\Tbb' = (T, \{B'_v : v \in T\})$ where for any $v \in T$, $B'_v = B_v \cup \{\bar{u}\}$ if $u \in B_v$ and $B'_v = B_v$ otherwise. $\Tbb'$ is a tree decomposition of $\Gset'_\downarrow$, resulting in $\tw_\downarrow(\Gset') \leq 3$, since $\tw_\downarrow(\Gset) \leq 2$. 
\end{proof}

For temporal versions of twin-width, the following results show the impossibility of finding meta-theorems for \MSO{} Model-Checking on temporal graphs of bounded temporal twin-widths.
We start with observations on temporal twin-widths. 
\begin{lemma}
\label{snapshottww_smaller_extendedtww}
Every temporal graph $\mathcal{G}$ satisfies $\mathrm{tww}_\infty(\mathcal{G}) \le \mathrm{tww}_\rightarrow(\mathcal{G})$.
\end{lemma}
\begin{proof} Let $0 \le t < \tau$.
The only thing to remark is that $\Gset_\rightarrow$ has $G_t$ as induced subgraph.
It is induced by vertex set $V_t = \{ (u,t) : u \in V \}$.
Since the twin-width of a graph is bigger than the twin-width of its induced subgraphs, the lemma follows.
\end{proof}

\begin{lemma}
\label{extended_twin-width_bounded}
Let $\Gset =(V, E_0, \ldots, E_{\tau-1})$ be a temporal graph such that $|E_t| \le k$ for each snapshot~$t$. Then, $\ttwww(\Gset)\leq k+3$.
\end{lemma}
\begin{proof} If we let $V = \{u_1, \ldots, u_n\}$, then $V(\Gset_\rightarrow) = \{(u_i, t) \mid 1\leq i\leq n, t \in \llbracket0, \tau-1\rrbracket\}$. We will contract vertices in the following order : first we contract $(u_1,0)$ with $(u_2, 0)$ forming $X_1^{(0)}$, then $(u_1, 1)$ and $(u_2, 1)$ forming $X_1^{(1)}$, and so on until $(u_1, \tau-1)$ and $(u_2, \tau-1)$ forming $X_1^{(\tau-1)}$. Afterwards, we contract $X_1^{(0)}$ with $(u_3,0)$, then $X_1^{(1)}$ and $(u_3, 1)$, and so on until $X_1^{(\tau-1)}$ and $(u_3, \tau-1)$. By continuing like this until contracting all vertices $X_{n-2}^{(t)}$ and $(u_{n}, t)$, for all $0\leq t<\tau$, we will obtain a line graph, which is of twin-width bounded by one. The red degree is bounded by $k + 3$ throughout the whole process. In fact, for $0 < t < \tau -1$ and $1 \le i \le n-1$, $X_i^{(t)}$ will have red edges to $X_i^{(t-1)}, X_{i-1}^{(t+1)}, (u_{i+1}, t+1)$ and to no more than to $k$ vertices $(u_j, t)$ with $j > i + 1$, as there are at most $k$ edges in $E_t$ by hypothesis.
Furthermore, for $t=0$ (resp.\ $t=\tau-1$), $X_i^{(t)}$ has only two red edges to $X_{i-1}^{(t+1)}, (u_{i+1}, t+1)$ (resp.\ one red edge to $X_i^{(t-1)}$), and $k$ red edges to vertices on the same snapshot. We have shown that $\mathrm{tww}_\rightarrow(\mathcal{G}) \le k + 3$.
\end{proof}

\begin{theorem}
\label{absence_meta_extended}
Assuming $P \neq NP$, there is no FPT algorithm parameterized by the combined value (or the sum) of the expanded twin-width $\mathrm{tww}_\rightarrow$ of the input temporal graph and the size of the input formula for \MSOone{} model checking.
\end{theorem}
\begin{proof} Like in the proof of Theorem~\ref{th:neg_tw_inf}, we show that, if such an FPT algorithm existed, we could solve the $3$-colorability problem on static graphs in polynomial time. 
For $G=(V,\{e_1, \ldots, e_m\})$ an arbitrary static graph, we construct $\Gset_G = (V, \{e_1\}, \ldots, \{e_m\})$ in polynomial time in $||G||$. Since $(\Gset_G)_\downarrow=G$, we can decide $3$-colorability on
$(\Gset_G)_\downarrow$. Moreover, by Lemma~\ref{extended_twin-width_bounded}, $\ttwww(\Gset)\leq 4$. We use $\varphi_{\text{color}\downarrow}$ from proof of Theorem~\ref{th:neg_tw_inf}, that
ensures that $\structone{\Gset_G}\models\varphi_{\text{color}\downarrow}$ if and only if $(\Gset_G)_\downarrow$ is $3$-colorable. 
We could decide it in time $O(f(|\varphi|, 4) \cdot ||\mathcal{G}||^c) = O(||G||^c)$, a contradiction with $3$-colorability of static graphs being an NP-complete problem.
\end{proof}

Together with Lemma~\ref{snapshottww_smaller_extendedtww}, we get the following corollary.
\begin{corollary}\label{corollary_infinite_twinwidth}
Assuming $P \neq NP$, there is no FPT algorithm parameterized by the combined value (or the sum) of the infinite twin-width $\mathrm{tww}_\infty$ of the input temporal graph and the size of the input formula for \MSOone{} model checking.
\end{corollary}

There exists $f$ such that for any static graph $G$, $\tww{G}\leq f(\tw(G))$~\cite{JP22}, hence for any temporal graph $\Gset$,
$\twww_\downarrow(\Gset)\leq f(\tw_\downarrow(\Gset))$. Theorem~\ref{th:neg_rtb} has then the following corollary, that we also give below an independent proof.
Note that the result does not follow from Lemma~\ref{snapshottww_smaller_extendedtww} and Theorem~\ref{absence_meta_extended}.
In particular, we do not know how the underlying twin-width $\mathrm{tww}_\downarrow$ relates to the expanded twin-width $\mathrm{tww}_\rightarrow$.
\begin{theorem}\label{thm:fpt-underlying-twww}
Assuming $P \neq NP$, there is no FPT algorithm parameterized by the combined value (or the sum) of the underlying twin-width $\mathrm{tww}_\downarrow$ of the input temporal graph and the size of the input formula for \MSOone{} model checking.
\end{theorem}
\begin{proof} 
We proceed similarly as in the proof of Theorem \ref{absence_meta_extended} but for $G = (V,E)$, we construct $\Gset_G = (V, E, K_{|V|})$. Then, $(\Gset_G)_\downarrow = K_{|V|}$ the complete graph with $|V|$ vertices. So $\mathrm{tww}_\downarrow(\Gset_G) = 0$. We build $\varphi \in$~\MSOone{} such that $\structone{\Gset_G} \models \varphi$ if and only if the first snapshot of 
$\Gset$ is 3-colorable.
First, we define $\psi(X) = \exists s\exists s'.~(s\in X\wedge s'\in X \wedge \forall s''.~(\neg \mathrm{succ}(s'',s) \wedge \neg \mathrm{succ}(s'', s')) \wedge \mathrm{adj}(s, s'))$ that states that there exist two elements in the set $X$ that are adjacent and in the first snapshot.
Then $\varphi = \exists X_1\exists X_2\exists X_3.~(\texttt{Partition}(X_1, X_2, X_3) \wedge \neg \psi(X_1) \wedge \neg \psi(X_2) \wedge \neg \psi(X_3))$.
\end{proof}

\section{$\FOone^\Delta$ Model Checking for bounded Expanded Twin-Width}
\label{section.dtw}
We relax $\MSOT$ down to First Order and aim at proving the following result.
\setcounter{toutpetit}{\value{theorem}}
\begin{theorem}\label{th:amt-tww}
Given $\Gset=(V,E_0,\dots, E_{\tau-1})$ a temporal graph such that $\ttwww(\Gset)\leq d$, $\varphi$ a formula in $\FOone^\Delta$, and a $d$-contraction sequence of $\expG{\Gset}$ (called decomposition for twin-width), one can decide if $\graphstruct{\Gset}_\Delta\models\varphi$
in $O(||\Gset|| \cdot f(|\varphi|, d))$, with $f$ a computable function that depends on $|\varphi|$ and $d$. 
\end{theorem}

First Order logic over temporal graphs (\FOone) is a restriction of \MSOone, in which we only consider variables from $\mathcal{X}_0$.
Given $\Delta \in \N \setminus \{0\}$, the $\Delta$-variant of a logic over temporal graphs has the same syntax except that the predicate $\app(x, x')$ is replaced by $\simd_\Delta(x, x')$, which holds when the time vertices $x$ and $x'$ represent the same untime vertex, and appear in two snapshots whose distance is at most $\Delta$.
The resulting signature is $\sigma_\Delta=\{\adj,\suc,\simd_\Delta\}$, and we define $\FOone^\Delta=\FO(\sigma_\Delta)$.
To a temporal graph  $\Gset=(V,E_0,\dots, E_{\tau-1})$ we associate the $\sigma_\Delta$-structure $\graphstruct{\Gset}_\Delta$ with domain $D_{\graphstruct{\Gset}_\Delta}=V\times\llbracket 0, \dots, \tau-1\rrbracket$. The relations $\adj$ and $\suc$ are interpreted as in $\MSOone$, and $\simd_\Delta^{\graphstruct{\Gset}_\Delta}=\{((v,t),(v,t'))\mid v\in V, 0\leq t,t'\leq \tau-1, \textrm{ and }|t-t'|\leq \Delta\}$. The predicate $\simd_\Delta$ does not add expressivity to the fragment, since it can be expressed with $\Delta$ existentially introduced variables and the predicate $\suc$. We introduce it as syntactic sugar, and to shorten the formulas.
We also consider the smaller signature $\sigma=\{\adj,\suc\}$ and write  $\FO=\FO(\sigma)$. The $\sigma$-structure associated with $\Gset$ is obtained from $\graphstruct{\Gset}_\Delta$ by omitting the interpretation of $\simd_\Delta$; we denote it by $\graphstruct{\Gset}$.

We show expressivity of $\FOone^\Delta$ in Section~\ref{subsec:mso-ex}, with a formula expressing the existence of a $\Delta$-clique of size $k$ in a temporal graph, for a given $\Delta$. A set of vertices is a $\Delta$-clique if there exists a point in time where every pair of its elements is connected by
an edge \emph{within a timeframe of $\Delta$-snapshots}.
Back to Model Checking, the following lemma is straightforward.
\begin{lemma}\label{lem:fo-fodelta}
For any $\varphi$ a formula in $\FOone^\Delta$, we can build a formula $\varphi'$ in $\FO$ such that for any $\sigma_\Delta$-structure $\Aset$, 
$\Aset\models\varphi$ if and only if $\Aset\models\varphi'$.
Moreover, $|\varphi'|\leq|\varphi|\cdot\Delta$.
\end{lemma}

Since we address logical structures, a proper definition of their twin-width is needed.
Given a temporal graph $\Gset=(V,E_0,\dots, E_{\tau-1})$, a time vertex $a\in V\times \llbracket 0, \tau-1\rrbracket$, and a strictly positive natural number $r$, 
$\localstatic{r}(a)$ is a logical structure, more general than a graph. We first need to define the twin-width of our logical structures. 
Let $\graphstruct{\Gset}$ be a $\sigma$-structure. In a contraction sequence $G_n,\dots, G_1$ of $\graphstruct{\Gset}$, 
each $G_i=(V_i,E_i^{\adj}\uplus E_i^{\suc}\uplus R_i)$ is now a structure over the signature $\{\suc,\adj\}$, along with $R_i$ the set of edges relating two non-homogeneous vertices:
In a contraction replacing a 
pair of vertices $u,v\in G_{i+1}$ by $w$, for any other vertex $x$ in $G_i$, $wx\in E_i^{\adj}$ if $ux, vx\in E_{i+1}^{\adj}$ and $wx\in E_i^{\suc}$ if $ux, vx\in E_{i+1}^{\suc}$. Otherwise, $wx\in R_i$. Again,
$wx\notin E_i^{\suc}\cup E_i^{\adj}\cup R_i$ if both $ux, vx\notin E_i^{\suc}\cup E_i^{\adj}\cup R_i$ and in every other cases, $wx\in R_i$. This means that, when contracting two vertices, that are both
related to the same vertex $x$, but \emph{not with the same relation}, the new vertex is related to $x$ with a red edge. 

Since $\ttwww(\Gset)= \twww(\expG{\Gset})$ and $\Gset_{\rightarrow}$ does not distinguish between internal snapshot edges and edges between two successive snapshots, we obtain $\ttwww(\Gset)\leq \twww(\graphstruct{\Gset})$. 
In reality, because of its specific structure, the increased twin-width in the static expansion graph is very limited.
The following property is crucial in proving Theorem~\ref{th:amt-tww}.
\begin{theorem}\label{th:tww-structure}
Let $\Gset=(V,E_0,\dots, E_{\tau-1})$ be a temporal graph. Then $\ttwww(\Gset)\leq \tww{\graphstruct{\Gset}}\leq \ttwww(\Gset)+2$, and any $d$-contraction of
$\Gset$ gives a $d$+2-contraction of $\graphstruct{\Gset}$.
\end{theorem}
\begin{figure}
    \centering 
      \begin{tikzpicture}[
        arrow/.style={->, dashed, shorten >=3pt, shorten <=3pt, >=stealth},
        element/.style={inner sep=1pt, font=\footnotesize},
        container/.style={draw, circle, fill=white, minimum size=1.8cm}, 
        thick_edge/.style={thick},
        container_edge/.style={red, ultra thick, opacity=0.5},
        node distance=0.5cm
    ]

        \node[container] (U) at (0,0) {};
        \node[left=0.2cm of U, font=\small] {$u$}; 
        
        \node[element] (u1_t1) at (0, 0.25) {$(u_1, t_1)$};
        \node[element] (u2_t1_minus) at (0, -0.25) {$(u_2, t_1-1)$};

        \node[container] (W1) at (5.0, 2.2) {};
        \node[right=0.2cm of W1, font=\small] {$w_1$};
        \node[element] (u1_t1_minus) at (5.0, 2.2) {$(u_1, t_1-1)$};

        \node[container] (W2) at (5.0, 0) {};
        \node[right=0.2cm of W2, font=\small] {$w_2$};
        \node[element] (u2_t1) at (5.0, 0) {$(u_2, t_1)$};

        \node[container] (W3) at (5.0, -2.2) {};
        \node[right=0.2cm of W3, font=\small] {$w_3$};
        \node[element] (u1_t1_plus) at (5.0, -2.2) {$(u_1, t_1+1)$};

        \begin{scope}[on background layer]
            \draw[container_edge] (U.center) to[bend right=15] (W3.center);
        \end{scope}

        
        \draw[arrow] (u1_t1_minus) -- (u1_t1);
        \draw[arrow] (u2_t1_minus) -- (u2_t1);
        
        \draw[arrow] (u1_t1) to[bend right=10] (u1_t1_plus);

        \draw[thick_edge] (u1_t1_minus) -- (u2_t1_minus);
        \draw[thick_edge] (u1_t1) -- (u2_t1);

    \end{tikzpicture}
   \caption{The nodes $u(\Gset)$, $w_1(\Gset)$, $w_2(\Gset)$ and $w_3(\Gset)$. Thick black edges show the relation $\adj$, and dashed arrows the relation $\suc$. Red edge shows
   non homogeneity in the contraction sequence of $\expG{\Gset}$.} \label{fig:twinwidth}
    \end{figure}
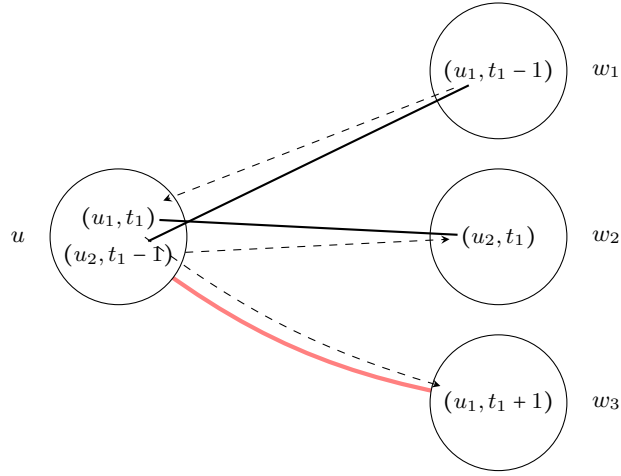

\begin{proof}
Let $d$ be a natural number, and $G_n=(V_n, B_n, R_n)$,\dots, $G_1=(V_1, B_1, R_1)$ a $d$-contraction sequence for $\expG{\Gset}$. We define $G'_n=(V_n, {B'}_n^{\adj}, {B'}_n^{\suc}, R'_n),\dots, G'_1=(V_1,{B'}_1^{\adj}, {B'}_1^{\suc}, R'_1)$ the contraction sequence for $\graphstruct{\Gset}$ following the same contractions as $G_n,\dots, G_1$.
It is clear that $|\{v\in V_i\mid xv\in R_i\}|\leq |\{v\in V_i\mid xv\in R'_i\}|$. 
We show that, for all $1\leq i\leq n$, for all $x\in V_i$, $|\{v\in V_i\mid xv\in R'_i\}|\leq |\{v\in V_i\mid xv\in R_i\}|+2$, i.e., no more than two edges incident to $x$ are black in $G_i$ and red in $G'_i$. 
For $v\in V_i$, we let $v(\Gset)$ denote the set of vertices of $V$ that are contracted in $v$ in $G_i$ (or in $G'_i$). 

We first observe that, if $vw\in B_i$ and $vw\in R'_i$, then there is $(u,t)\in v(\Gset)$ and $(u,t')\in w(\Gset)$ with $t'\in \{t-1,t+1\}$. Indeed, if $vw\in B_i$, then either there is no edge between $v(\Gset)$
and $w(\Gset)$, but in that case $vw\notin R'_i$, or every vertex of $v(\Gset)$ is linked to a vertex of $w(\Gset)$ in $\expG{\Gset}$. In that case, if there is no $(u,t)\in v(\Gset)$ and $(u,t')\in w(\Gset)$ with $t'\in \{t-1,t+1\}$, it implies that there is $0\leq t\leq \tau-1$ such that $v(\Gset)\subseteq V\times \{t\}$ and $w(\Gset)\subseteq V\times\{t\}$, i.e. all the time vertices in $v(\Gset)$ and $w(\Gset)$
belong to the same snapshot (otherwise, since two time vertices in two snapshots that are not successive are never linked, we would have vertices in $v(\Gset)$ not linked to some vertices of $w(\Gset)$).
If all the vertices in $v(\Gset)$ and $w(\Gset)$ are in the same snapshots, and are all linked together, they are all linked by the relation $\adj$, and $vw\in {B'}_i^{\adj}$, contradicting the fact that $vw\in R'_i$. 

Now let $u\in V_i$ and suppose there exists $w_1, w_2, w_3\in V_i$ all distinct and such that $uw_1, uw_2, uw_3\in B_i$ and $uw_1, uw_2, uw_3\in R'_i$. Then, following the previous observation,
there exist $(u_1,t_1), (u_2,t_2), (u_3,t_3)\in u(\Gset)$ and $(u_1,t'_1), (u_2,t'_2), (u_3, t'_3)\in w(\Gset)$ with $t'_i\in\{t_1-1, t_i+1\}$ for $i\in \{1,2,3\}$. 

\begin{itemize}
\item If $(u_1,t_1)=(u_2,t_2)=(u_3,t_3)$. Suppose, wlog, that $(u_1,t_1-1)\in w_1(\Gset)$ (by the observation above). Then, $(u_1, t_1+1)\in w_2(\Gset)$, because we know that $(u_1, t'_1)\in w_2(\Gset)$, and $(u_1, t_1-1)\in w_1(\Gset)$ (by construction, $w_1(\Gset)\cap w_2(\Gset)=\emptyset$). In that case, we should have $(u_1, t_1-1)\in w_3(\Gset)$ or $(u_1, t_1+1)\in w_3(\Gset)$,
but since $w_1, w_2$ and $w_3$ are all distinct, none of these possibilities is true. Hence, we reach a contradiction, and $(u_1,t_1), (u_2,t_2)$ and $(u_3,t_3)$ are not all equal.

\item Assume then, wlog, that $(u_1,t_1)\neq (u_2,t_2)$.

\begin{itemize}
\item If $t_1=t_2$, then $u_1\neq u_2$, since $(u_1, t')\in w_1(\Gset)$ with $t'\in\{t_1-1, t_1+1\}$, suppose wlog that $(u_1, t_1-1)\in w_1(\Gset)$. Then neither $((u_2,t_2),(u_1, t_1-1))$
nor $((u_1,t_1-1),(u_2, t_2))$ belong to
$\adjE\cup\tempE$ in $\expG{\Gset}$ because $(u_2,t_2)$ and $(u_1, t_1-1)$ being in two different snapshots, they cannot be adjacent, and $u_1\neq u_2$ implies that
$(u_2,t_2)$ is not a time successor of $(u_1,t_1-1)$. Then there is a vertex in $u(\Gset)$ and another in $w_1(\Gset)$ that are not linked in $\expG{\Gset}$, a contradiction with $uw_1\in B_i$. 

\item If $u_1=u_2$ (then $t_1\neq t_2$). Again, suppose wlog that $(u_1, t_1-1)\in w_1(\Gset)$. Since $uw_1\in B_i$, then there is an edge from $(u_2,t_2)=(u_1, t_2)$ to $(u_1, t_1-1)$. As we know that
$t_2\neq t_1$, either $t_2=t_1-1$, but it is impossible since $(u_1, t_1-1)$ is in $w_1(\Gset)$, it cannot be in $u(\Gset)$, or $t_2=t_1-2$ and $(u_2,t_2)=(u_1, t_1-2)$ and $(u_1, t_1-1)$ are linked by 
$\tempE$ in $\expG{\Gset}$.  Then, $(u_2, t'_2)\in w_2(\Gset)$, with $t'_2\in\{t_2-1, t_2+1\}=\{t_1-3, t_1-1\}$. But this again leads to a contradiction, because if $(u_2, t_1-3)\in w_2(\Gset)$, $(u_1,t_1)$
and $(u_2, t_1-3)$ are not linked in $\expG{\Gset}$ and this contradicts $uw_2\in B_i$, and if $(u_2, t_1-1)\in w_2(\Gset)$, since $u_2=u_1$, we have $(u_1, t_1-1)\in w_1(\Gset)\cap w_2(\Gset)$, which is impossible.  

\item If $u_1\neq u_2$ and $t_1\neq t_2$, we suppose that $(u_1, t_1-1)\in w_1(\Gset)$. If $t_2\neq t_1-1$, then there is no edge between $(u_1, t_1-1)$ and $(u_2, t_2)$, hence $uw_1\notin B_i$. 
Then $t_2=t_1-1$.  Now, $(u_2, t'_2)\in w_2(\Gset)$, with $t'_2\in \{t_2-1, t_2+1\}$. If $t'_2=t_2-1$, we have that $(u_2, t_1-2)\in w_2(\Gset)$. But this implies that there is no edge between 
$(u_1,t_1)$ and $(u_2, t_1-2)$ in $\expG{\Gset}$, a contradiction with the assumption that $uw_2\in B_i$. Then $t'_2=t_2+1=t_1$ and $(u_2, t_1)\in w_2(\Gset)$. 
\begin{itemize}
\item If $(u_3,t_3)\neq (u_2,t_2)$ and 
$(u_3,t_3)\neq (u_1, t_1)$, 
if $u_3\notin\{u_1,u_2\}$, it cannot be linked in $\expG{\Gset}$ both with $(u_1, t_1-1)$ and $(u_2, t_1)$ because these two vertices are in two different snapshots, and if $u_3\notin\{u_1,u_2\}$
the only possible edge between $(u_3,t_3)$ and $(u_1, t_1-1)$ and $(u_2, t_1)$ are adjacent edges in the same snapshot. This contradicts then the assumption that $uw_1$ and $uw_2$ are in $B_i$.
Then $u_3\in\{u_1,u_2\}$ and $t_3\in \{t_1, t_1-1\}$, because it is the only possibility for $(u_3,t_3)$ to be linked to both $(u_1, t_1-1)$ and $(u_2, t_2+1)=(u_2,t_1)$ in $\expG{\Gset}$. A case analysis
show that there is no possibility for both $u_3\in\{u_1, u_2\}$ and $t_3\in \{t_1, t_1-1\}$. 
\item Then $(u_3,t_3)=(u_1,t_1)$ or $(u_3,t_3)=(u_2,t_2)$. If $(u_3,t_3)=(u_1,t_1)$, then, by the observation made at the beginning of the proof, $(u_1, t_1+1)\in w_3(\Gset)$. But then 
$(u_2,t_2)=(u_2,t_1-1)$ is not linked to $(u_1, t_1+1)$ in $\expG{\Gset}$, which means that $uw_3\in R_i$, contradicting our assumption. This is illustrated on Figure~\ref{fig:twinwidth}. If $(u_3,t_3)=(u_2,t_2)$, similarly, $(u_2, t_2-1)=(u_2, t_1-2)\in 
w_3(\Gset)$ and $(u_1, t_1)$ could not be linked with $(u_2, t_1-2)$, hence $uw_3\notin B_i$.
\end{itemize}
\end{itemize} 
\end{itemize}
Each case yielded a contradiction, hence we can conclude that one cannot find $w_1,w_2,  w_3$, three distinct vertices in $V_i$ such that $uw_1, uw_2$ and $uw_3\in B_i$, and $uw_1, uw_2$ and
$uw_3\in R'_i$. Hence, any $d$-contraction sequence in $\expG{\Gset}$ is a $(d+2)$-contraction sequence in $\graphstruct{\Gset}$. 
\end{proof}

\begin{proof}[Proof of Theorem~\ref{th:amt-tww}]
Lemma~\ref{lem:fo-fodelta} allows to turn $\varphi$ into $\varphi'$, a formula over the signature $\{\suc,\adj\}$, that we turn into $\varphi''$ in prenex
normal form. Theorem~\ref{th:tww-structure} allows to 
turn the $d$-contraction sequence of $\expG{\Gset}$ into a $(d+2)$-contraction sequence of $\graphstruct{\Gset}$.
Now, Theorem~7.1 of Ref.~\cite{BKTW22} establishes that one can decide if $\graphstruct{\Gset}\models\varphi''$ in $O(\tau\cdot |V|\cdot g(|\varphi''|,d+2))$ for some computable function $g$, which means that one
can decide if $\graphstruct{\Gset}_\Delta\models\varphi$
in $O(||\Gset|| \cdot f(|\varphi|, d))$, with $f$ a computable function that depends on $|\varphi|$ and $d$.
\end{proof}

\subsection*{Derivative of Temporal Graphs}
\label{subsection.derivative}
We have shown for $\MSOone$ and $\FOone^\Delta$ that they are tractable on input a temporal graph $\Gset$ of bounded expanded tree-width $\tw_\rightarrow(\Gset)=\tw(\Gset_\rightarrow)$ or bounded expanded twin-width $\twww_\rightarrow(\Gset)=\twww(\Gset_\rightarrow)$, respectively.
While the latter result is a weak AMT (a $d$-contraction is required), the result on bounded expanded tree-width is a typical AMT, using \textit{e.g.}~\cite{RS95,B96,OS06}.
This is in contrast to many other versions of temporal width.
Even though the expanded tree-width has value still in $O(n)$ (precisely, at most $2n$) it is defined over $\Gset_\rightarrow$, a very large graph.
In particular, it is very easy to result in a grid as minor of $\Gset_\rightarrow$, as long as the snapshot graphs contain long paths.
Striving to control this kind of situation, we focus on smaller slices of $\Gset_\rightarrow$, called differentials, and examine their tree-widths instead.
Doing so, we bound the potential grid-like structure, while maintaining relevant information on the evolution of the structure of the graph with time.
We also consider the impact of the twin-width of differentials on AMT.
\begin{definition}[Differential]
\emph{Let $\mathcal G=(G_t)_{0\leq t<\tau}$ be a temporal graph, $\Delta$ an integer and $t$ a time instant when $0\leq t\leq\tau-\Delta$.
The \textit{differential of $\mathcal G$ at $t$ over $\Delta$} is defined as $\mathcal G^{t,\Delta}_\rightarrow$ where $\mathcal G^{t,\Delta}=(G_{x})_{t\leq x\leq t+\Delta-1}$.
It is precisely the static expansion graph of the $\Delta$ snapshots of $\mathcal G$ starting at $t$ and ending at $t+\Delta-1$.}
\end{definition}

For temporal graph $\mathcal G$ we define $\dtw_\Delta(\mathcal G)$  and $\dtww_\Delta(\mathcal G)$, respectively the \textit{$\Delta$-differential tree-width  and $\Delta$-differential twin-width of $\mathcal G$ over $\Delta$}, as the maximum tree-width, respectively maximum twin-width, of the differentials of $\mathcal G$ at any time $t$ over $\Delta$.

We describe in the sequel the Gaifman Locality property and use it to prove our last meta-theorems.
Recall that the \emph{Gaifman graph} of a temporal graph $\Gset=(V,E_0,\dots, E_{\tau-1})$ with respect to $\FOone^\Delta$ is
the graph $\Gaif(\graphstruct{\Gset}_\Delta)=(V\times \llbracket 0, \tau -1\rrbracket, E_\Gset)$ with
 $E_\Gset=\{(u,t)(v,t)\mid uv\in E_t\}\cup\{(u,t)(u,t+1)\mid u\in V, 0\leq t <\tau - 1\}\cup
\{(u,t)(u,t+\delta)\mid u\in V, 0\leq t < \tau-\delta,1\leq\delta\leq\Delta\}$.

Given $v,v'\in V_\Gset$, we let $\dist^\Gset(v,v')$ be the shortest path distance from $v$ to $v'$ in $\Gaif(\graphstruct{\Gset}_\Delta)$ if such a path exists, and $\dist^\Gset(v,v')=\infty$ otherwise.
Let $r\in\mathbb{N}$ be an integer and $a\in V\times\llbracket 0, \tau-1\rrbracket$ a time vertex.
The $r$-neighbourhood of $a$ is defined as $\neigh{r}(a)=\{b\in V\times\llbracket0,\tau-1\rrbracket\mid \dist^{\Gset}(a,b)\leq r\}$.
It contains time vertices in a ball of radius $r$ around $a$.
Figure~\ref{fig:neighbourood} exemplifies $r$-neighbourhood.
Its formalism allows the definition of an $r$-local formula from $\FOone^\Delta$: it is a formula whose truth value depends only on the substructure induced by the $r$-neighbourhood around $a$.
\begin{definition}[$r$-local formula]
\emph{A formula $\varphi(x)$  is \emph{$r$-local} if for every
structure $\Aset$ over a domain $D$ and $a \in D$,
$\Aset \models_{\nu_0[x\leftarrow a]} \varphi(x) \quad \text{iff} \quad \struct{\neigh{r}(a)}^{\Aset} \models_{\nu_0[x\leftarrow a]} \varphi(x)$.}
\end{definition}
\begin{definition}[basic $r$-local sentence]
\emph{A \emph{basic $r$-local sentence} is a first-order formula of the following form:
$\exists x_1\exists x_2 \ldots \exists x_k.
\displaystyle\bigwedge_{1 \leq i < j \leq k} \delta_{2r}(x_i, x_j) \land
\displaystyle\bigwedge_{i = 1}^k \vartheta(x_i)$,
where for all $i$, $\vartheta(x_i)$ is a first-order $r$-local formula
and $\delta_{2r}(x,y)$ is a first-order formula s.t. $\Aset\models_\nu \delta_{2r}(x,y)$
if and only if the distance between $\nu(x)$ and $\nu(y)$ in the Gaifman graph of $\Aset$
is strictly greater than $2r$.}
\end{definition}
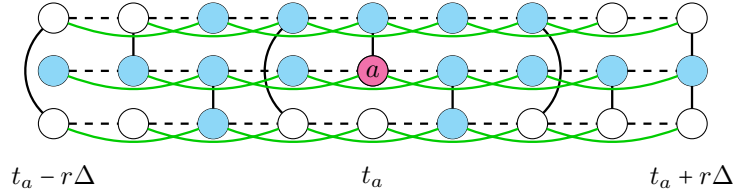
\begin{figure}
    \centering
    \def\hsep{1.5}
    \def\vsep{1}
    \begin{tikzpicture}[scale=0.7, dot/.style={draw,circle,minimum size=0.4cm,inner sep=0pt,outer sep=0pt}, voisin/.style={fill=cyan!40,circle,minimum size=0.4cm,inner sep=0pt,outer sep=0pt}]
        \path \foreach \i in {0,...,2}
				\foreach \j in {0,...,8}
                        {coordinate [dot] (u\i\j) at (\hsep*\j,\vsep*-\i)};

        \node[draw,circle,minimum size=0.4cm,inner sep=0pt,outer sep=0pt,fill=magenta!70]
            (v14) at (\hsep*4, \vsep*-1) {$a$};
        

        \node[voisin] (v10) at (\hsep*0, -\vsep*1) {};

        \node[voisin] (v11) at (\hsep*1, -\vsep*1) {};

        \node[voisin] (v02) at (\hsep*2, -\vsep*0) {};
        \node[voisin] (v12) at (\hsep*2, -\vsep*1) {};
        \node[voisin] (v22) at (\hsep*2, -\vsep*2) {};

        \node[voisin] (v03) at (\hsep*3, -\vsep*0) {};
        \node[voisin] (v13) at (\hsep*3, -\vsep*1) {};

        \node[voisin] (v04) at (\hsep*4, -\vsep*0) {};

        \node[voisin] (v05) at (\hsep*5, -\vsep*0) {};
        \node[voisin] (v15) at (\hsep*5, -\vsep*1) {};
        \node[voisin] (v25) at (\hsep*5, -\vsep*2) {};

        \node[voisin] (v06) at (\hsep*6, -\vsep*0) {};
        \node[voisin] (v16) at (\hsep*6, -\vsep*1) {};

        \node[voisin] (v17) at (\hsep*7, -\vsep*1) {};

        \node[voisin] (v18) at (\hsep*8, -\vsep*1) {};

		\draw[line width=0.2ex, black]
        
        (u00) to[bend right=45] (u20)
		(u01)--(u11)
		(u14) -- (u04) 
        (u12) -- (u22)
        (u23) to[bend left=45] (u03)
        
        (u17)--(u27)
        (u15)--(u25)

        (u06) to[bend left=45] (u26)
        
        (u08) -- (u18) -- (u28);    


        \draw[dashed, line width=0.2ex]
    
        \foreach \i in {0, 1, 2}
        \foreach \j in {0,...,7}
        {(u\i\j)--(u\i\the\numexpr\j+1)};

        \draw[green!80!black, line width=0.2ex]
            \foreach \i in {0, 1, 2}
                \foreach \j in {0,...,6}
                    {(u\i\j) to[bend right=20] (u\i\the\numexpr\j+2)};

        \node (tinf) at (0, \vsep*-3) {$t_a - r \Delta$};
        \node (tmax) at (12.0, \vsep*-3) {$t_a + r \Delta$};
        \node (ts) at (6.0, \vsep*-3) {$t_a$};
	\end{tikzpicture}
    \caption{(blue vertices) The $r$-neighbourhood $\neigh{r}(a)$ of $a$, for $\Delta = 2$ and $r = 2$;
    (black edges) The edges of the original graph $\Gset$;
    (dashed edges) The edges which represent $\suc$;
    (green edges) The edges which represent $\simd_\Delta$. In $G^r(a)$, green edges are not considered.}
    \label{fig:neighbourood}
\end{figure}

We will make use of the Gaifman locality theorem on first order logic:
\begin{theorem}[Gaifman Locality Theorem, \cite{G82,L04}]\label{th:gaifman}
Every first-order formula $\varphi(x)$ over a relational signature $\sigma$ is equivalent to a boolean combination of $r$-local formulas $\psi(x)$, and basic $r$-local sentences, with $r\leq 7^{\qr(\varphi)}$. Moreover, the transformation is effective. 
\end{theorem}

This locality property of first-order logic allows us to consider a local substructure of the original temporal graph. Given $\Gset=(V,E_0,\dots, E_{\tau-1})$ and $a=(v,t)$ with $v\in V$ and $0\leq t\leq \tau-1$, we define
$\localstatic{r}(a)=\struct{\neigh{r}(a)}^{\graphstruct{\Gset}}$, the substructure of $\graphstruct{\Gset}$ on the domain $\neigh{r}(a)$.

The following property about tree-width hints that we can use Courcelle's Theorem to obtain easily a meta-theorem for $\dtw_\Delta(\Gset)$. 
We aim at proving a stronger result, using twin-width instead.
If the temporal graph has a bounded $\Delta$-differential twin-width, for some $\Delta$,
given a temporal vertex $a$, the twin-width of
$\localstatic{r}(a)$, the substructure of $\graphstruct{\Gset}$ around the neighbourhood of $a$, is also strongly related to some differential twin-width of $\Gset$, as follows.
\begin{lemma}\label{lem:tw}
Given a temporal graph $\Gset=(V,E_0,\dots, E_{\tau-1})$, a time vertex $a\in V\times \llbracket 0, \tau-1\rrbracket$, and $r$ a strictly positive natural number, 
$\tw(\localstatic{r}(a))\leq \dtw_{2r\Delta+1}(\Gset)$, $\tww{\localstatic{r}(a)}\leq \dtww_{2r\Delta+1}(\Gset)+2$, and $|V(\localstatic{r}(a))|\leq |V| \cdot (2r\Delta+1)$.
\end{lemma}

\begin{proof}
We start by proving that $|V(\localstatic{r}(a))|\leq |V| \cdot (2r\Delta+1)$. By definition, the nodes in $V(\localstatic{r}(a))$ are those in $\neigh{r}(a)\subseteq V\times\llbracket 0, \tau-1\rrbracket$. So obviously, $|V(\localstatic{r}(a))|\leq |V| \cdot \tau$. However, for all $b\in\neigh{r}(a)$, it must be that $\dist^\Gset(a,b)\leq r$. Then, the furthest snapshot to which $b$ can belong
to can be reached by a sequence of $r$ $\simd_\Delta$ edges in the Gaifman graph. If $a=(v,t)$, it means that $b = (v, t + \delta)$ with $-r\Delta \le \delta \le r\Delta$
Then, if $a=(v,t)$, 
$\neigh{r}(a)\subseteq V\times \llbracket t-r\Delta, t+r\Delta\rrbracket$, hence $|V(\localstatic{r}(a))|\leq |V| \cdot (2r\Delta+1)$.

Consider now $\Gset^{t_a - r \Delta,2r\Delta + 1}_\rightarrow$, the differential of $\Gset$ at $t_a-r\Delta$. It follows from the definitions that $\Gaif(\localstatic{r}(a))$ is an induced subgraph of 
$\Gset^{t_a - r \Delta,2r\Delta + 1}_\rightarrow$, hence $\tw(\localstatic{r}(a))\leq \tw(\Gset^{t_a - r \Delta,2r\Delta + 1}_\rightarrow)$. It follows that $\tw(\localstatic{r}(a)) \leq \dtw_{2r\Delta +1}(\Gset)$. Moreover, $\localstatic{r}(a)$ is an induced substructure of $\graphstruct{\Gset^{t_a - r \Delta,2r\Delta + 1}}$, 
hence $\tww{\localstatic{r}(a)}\leq \tww{\graphstruct{\Gset^{t_a - r \Delta,2r\Delta + 1}}}\leq \tww{\Gset^{t_a - r \Delta,2r\Delta + 1}}+2\leq \dtww_{2r\Delta+1}(\Gset)+2$, by Theorem~\ref{th:tww-structure}.
\end{proof}

For each $\varphi\in \FOone^\Delta$, we let $\Delta_\varphi=2r\Delta+ 1$, with $r$ the integer coming from
Gaifman Locality Theorem (Theorem~\ref{th:gaifman}).

\begin{theorem}\label{th:fot1ttwww}
Given a $\FOone^\Delta$ formula $\varphi$, a temporal graph $\Gset$ such that there is some $d$ such that
$\dtww_{\Delta_\varphi}(\Gset)\leq d$ and given, for each differential $\mathcal G^{t,\Delta_\varphi}_\rightarrow$, a $d$-contraction sequence, 
we can decide if
	$\graphstruct{\Gset}_\Delta \models \varphi$ in
	$O(|\Gset|^3\cdot f(\Delta, d,
	|\varphi|))$, with $f$ a computable
	function only depending on $\Delta,	|\varphi| \text{ and } \allowbreak d$.
	\end{theorem}
	
	\begin{proof}
	We first decompose $\varphi$ in basic $r$-local sentences, with $r\leq 7^{\qr(\varphi)}$. We can decide whether $\graphstruct{\Gset}_\Delta$ is a model
	of each of this basic $r$-local sentence in the following way (detailed right below): for each $r$-local formula $\vartheta(x)$, for each $a=(v_a,t_a)\in V\times \llbracket 0, \tau-1\rrbracket$, we compute $\localstatic{r}(a)$. 
	By Lemma~\ref{lem:tw}, $\tww{\localstatic{r}(a)}\leq \dtww_{\Delta_\varphi}(\Gset)+2$.
	Moreover, Theorem~\ref{th:tww-structure} states that the $d$-contraction 
	sequence for $\Gset^{t_a - r \Delta,\Delta_\varphi}_\rightarrow$ gives a ($d$+2)-contraction sequence for $\graphstruct{\Gset^{t_a - r \Delta,\Delta_\varphi}}$,
	from which we can deduce a $d$+2-contraction sequence for its substructure $\localstatic{r}(a)$. Then, the transformation of $\vartheta$ into $\vartheta'$ via Lemma~\ref{lem:fo-fodelta} and Theorem 7.1 of~\cite{BKTW22}
	allows us to conclude if $\localstatic{r}(a)\models\vartheta(a)$ in $O(||\localstatic{r}(a)||\cdot f_a(\Delta\cdot|\vartheta(a)|,d+2))$. It remains to find $k$ elements $a_1, \ldots, a_k$ that are pairwise distant from more than $2r + 1$ to determine the truth value or the basic $r$-local sentence. It can be done in cubic time through iterations of a graph search analysis. Boolean combinations of the basic $r$-local formulas give the final answer.
	\end{proof}

Above Theorem~\ref{th:fot1ttwww} implies the same result when replacing twin-width with tree-width.
We also give an independent proof for the latter parameter, as follows.
In order to obtain an AMT for the $\Delta$-differential tree-width, we first show that the tree-width of $\localstatic{r}(a)$ is bounded by the $\Delta'$-differential tree-width of $\Gset$, for a $\Delta'$ we will
make explicit. We will then be able to use $\localstatic{r}(a)$ as input of the Model-Checking problem for MSO over graphs and apply Courcelle's theorem.

Lemma~\ref{lem:tw} paves the way to use Courcelle's theorem on $\localstatic{r}(a)$.

 \begin{lemma}\label{lem:local-sen}
	Given a basic $r$-local sentence in $\FOone^\Delta$,  $\varphi = \exists x_1\exists x_2 \ldots \exists x_k.~
\displaystyle\bigwedge_{1 \leq i < j \leq k} \delta_{2r}(x_i, x_j) \allowbreak \land
\displaystyle\bigwedge_{i = 1}^k \vartheta(x_i)$ and $\Gset = (V, E_0, \ldots, E_{\tau-1})$ a temporal graph,
we can decide if $\graphstruct{\Gset}_\Delta \models \varphi$ in $O(\allowbreak f(\Delta,\allowbreak \dtw_{2r\Delta + 1}(\Gset),
|\varphi|) \cdot |\Gset|^3)$, where $f$ is a calculable function
depending only on $\Delta$, $\dtw_{2r \Delta + 1}(\Gset)$ and $\varphi$.
\end{lemma}

\begin{proof}
	To prove this result, we propose the following algorithm: for every $a\in V \times \llbracket0, \tau -1\rrbracket$, 
	we compute $\localstatic{r}(a)$. We check whether $\graphstruct{\Gset}_\Delta\models\vartheta(a)$, which is equivalent to 
	$\struct{\neigh{r}(a)}^{\graphstruct{\Gset}_\Delta}\models\vartheta(a)$ by $r$-locality of $\vartheta$. By Lemma~\ref{lem:fo-fodelta},
	it is equivalent to $\struct{\neigh{r}(a)}^{\graphstruct{\Gset}}\models\vartheta'(a)$, i.e. $\localstatic{r}(a)\models\vartheta'(a)$.
	Thanks to Courcelle's theorem (Theorem ~\ref{th:courcelle}), this can be done in $O(|\localstatic{r}(a)| \cdot
		g_a(|\vartheta'(a)|, \tw(\localstatic{r}(a))))$
	for $g_a$ a calculable function depending only on $|\vartheta'(a)|$
	and $\tw(\localstatic{r}(a))$. By Lemmas~\ref{lem:fo-fodelta} and \ref{lem:tw}, it means that it can be done in 
	$O(|\localstatic{r}(a)| \cdot h_a(|\vartheta(a)|\cdot \Delta, \dtw_{2r\Delta+1}(\Gset)))$ for a computable function $h_a$.
	Hence, since for every $a\in V\times\llbracket 0,\tau-1\rrbracket$, $|V(\localstatic{r}(a))|\leq |V|\cdot 2r\Delta+1$, 
	the set $\llbracket \vartheta \rrbracket = \{a\in V \times \llbracket0, \tau -1\rrbracket\mid \graphstruct{\Gset}_\Delta\models\vartheta(a)\}$
	can be computed in time $O(|V|\cdot\tau\cdot |V|\cdot (2r\Delta+1)\cdot h(|\varphi|\cdot\Delta,  \dtw_{2r\Delta+1}(\Gset)))$, i.e.
	in time $O(|\Gset|^2\cdot h(|\varphi|\cdot\Delta,  \dtw_{2r\Delta+1}(\Gset)))$.
	For $\Gset$ to satisfy $\varphi$, we need to find $k$ elements from $\llbracket \vartheta \rrbracket$ that are pairwise distant
	from more than $2r + 1$ in $\Gaif(\Gset)$. To compute the distance between two elements $a$ and $b$ from $\llbracket \vartheta\rrbracket$, we can use 
	a graph search algorithm, hence it can be done in time $O(|V(\Gaif(\Gset))|+|E(\Gaif(\Gset))|)$, i.e. in $O(|\Gset|)$. Hence, 
	finding $k$ such elements can be done in time $O(k\cdot(|V|\cdot\tau)^2\cdot|\Gset|)$. 
	Combining the two steps of the algorithm, the time complexity is in $O(|\Gset|^3\cdot h(|\varphi|, \dtw_{2r\Delta+1}(\Gset)))$.
\end{proof}

We finally apply Theorem~\ref{th:gaifman} to obtain the following algorithmic meta-theorem: for $\varphi$ a formula in $\FOone^\Delta$,
we let $\Delta_\varphi=2r\Delta+1$, with $r\leq 7^{\qr(\varphi)}$ the integer from Theorem~\ref{th:gaifman}. 

\begin{theorem} \label{th:fot1ttw}
	Given a $\FOone^\Delta$ formula $\varphi$ and a temporal graph $\Gset$, we can decide if
	$\graphstruct{\Gset}_\Delta \models \varphi$ in
	$O(|\Gset|^3\cdot f(\Delta, \dtw_{\Delta_\varphi}(\Gset),
	|\varphi|))$, with $f$ a computable
	function only depending on $\Delta,	|\varphi| \text{ and } \allowbreak
	\dtw_{\Delta_\varphi}(\Gset)$.
\end{theorem}
\begin{proof}
	Let $\varphi$ be a $\FOone^\Delta$ formula
	and $\Gset$ be a temporal graph. By Theorem~\ref{th:gaifman}, we can decompose
	$\varphi$ into an equivalent boolean combination of basic $r$-local sentences,
	such that $r \leq 7^{\qr(\varphi)}$.
	For each of these basic
	$r$-local sentences, Lemma~\ref{lem:local-sen}
	exhibits a fixed-parameter tractable algorithm to decide if $\Gset$ satisfies the formula. The size
	of the Boolean combination only depends on $\varphi$: 
	we can hence check if $\graphstruct{\Gset}_\Delta \models \varphi$ by combining the results of Lemma~\ref{lem:local-sen}, all in time in
	$O(|\Gset|^3\cdot f(\Delta, \dtw_{\Delta_\varphi}(\Gset),
	|\varphi|))$, with $f$ a computable
	function only depending on $\Delta,	|\varphi| \text{ and } \allowbreak
	\dtw_{\Delta_\varphi}(\Gset)$.
\end{proof}

\begin{remark}
Our definition of graphs with bounded $\Delta$-differential tree-width is somewhat reminiscent of the notion of locally bounded tree-width~\cite{FrickGrohe99, FG06}. 
However, our definition is tailored for temporal graphs and is not
captured by the notion of locally bounded tree-width. In particular, for every $\Delta$, it is possible to design a class of temporal graphs whose $\Delta$-differential tree-width is bounded but the class
of its static expansion graphs is not
locally bounded tree-width. Conversely, the class of temporal graphs in which each snapshot is an $n\times n$ grid has not a $\Delta$-differential tree-width bounded, for any $\Delta$, but the class
of its static expansion graphs has locally-bounded tree-width~\cite{SchmitzBuffiere25}.
\end{remark}

\section{Examples of properties expressed in $\MSOtwo$}\label{subsec:mso-ex}
We show in this section examples of classical properties that can be expressed on temporal graphs using \MSOtwo. 
 
\paragraph*{Temporal Matching}
Let $\Gset=(V,E_0\dots, E_{\tau-1})$ be a temporal graph. Two edges are \emph{dependent} if they are incident to a same time vertex. A \emph{temporal matching} is then a set of independent
edges, meaning that no two edges are incident to the same time vertex. The formula $\tempmatch(X):=\forall e\forall e'.~ (e\in X\wedge e'\in X\wedge e\neq e'\rightarrow \neg (\exists x.~ \inc(e,x)\wedge \inc(e',x)))$ expresses the fact that a set $X$ of temporal edges is a temporal matching. 

Sometimes, one may ask for a stronger requirement for two edges to be independent: they cannot be incident to two time vertices that represent the same untime vertex in two snapshots too close in time,
typically in an interval smaller than some given $\Delta$. The formula becomes then $\tempmatch_\Delta(X):=\forall e\forall e'.~ (e\in X\wedge e'\in X\wedge e\neq e'\rightarrow \neg (\exists x \exists x'.~\inc(e,x)\wedge \inc(e',x')\wedge \suc_{\leq\Delta}(x,x')))$.

A similar yet more involved notion has been described in~\cite{BB18,BBR20,MMNZZ20,MMNZZ23}. Let $\Gset=(V,E_0,\dots, E_{\tau-1})$ be a temporal graph, and $\gamma$ an integer, representing a minimum duration. A \emph{$\gamma$-edge} is a set of edges between two time vertices, starting at some instant $t$
and repeating $\gamma$ times, i.e., if $uv\in E_i$ for $t\leq i<t+\gamma$, then $\{(uv,t),\dots, (uv,t+\gamma-1)\}$ is a $\gamma$-edge in $\Gset$. Two $\gamma$-edges $e_1$ and $e_2$ are \emph{dependent} if there
exist a vertex $v$ and a snapshot $t$ such that $(i)$ there exists $u\in V$ and $(uv,t)\in e_1$ and $(ii)$ there exists $w\in V$ and $(wv,t)\in e_2$. Note that $u$ and $w$ are not necessarily distinct. 
A \emph{$(k,\gamma)$-temporal matching} is a set of $k$ independent $\gamma$-edges.

We express the fact that a set of edges is a $\gamma$-edge by identifying an ``initial'' edge in the set, one that appears in the earlier snapshot, and requiring that all the edges appearing in the set are 
incident to the same untime vertices, and appear within a timeframe of size $\gamma$, and finally, that all the edges in the required time interval actually exist and belong to the set.

Namely, 
\begin{align*}
\gammaedge(X)&:= \exists e_1\exists u_1\exists v_1.~(e_1\in X\wedge \inc(e_1,u_1)\wedge \inc(e_1,v_1))\\
&\wedge (\forall e.~ (e\in X\rightarrow (\exists u\exists v.\inc(e,u)\wedge\inc(e,v)\wedge \suc_{\leq\gamma}(u_1,u)\wedge \suc_{\leq\gamma}(v_1,v))))\\
&\wedge \forall u\forall v.~ (((\suc(u_1,u)\wedge \suc(v_1,v))\vee\bigvee_{i=2}^\gamma(\suc_i(u_1,u)\wedge\suc_i(v_1,v))))\rightarrow\\
 &\exists e.~(\inc(e,u)\wedge\inc(e,v)\wedge e\in X))
\end{align*}

We can then define $k-\tempmatch(X_1,\dots,X_k):= \bigwedge_{1\leq i\leq k} (\gammaedge(X_i)) \wedge \forall e\forall e'.~ ((\bigvee_{i\neq j} (e\in X_i\wedge e'\in X_j)) \rightarrow \neg\exists x.~(\inc(e,x)
\wedge \inc(e',x)))$

\paragraph*{Cut of size 1}
Given a temporal graph $\Gset=(V,E_0,\dots, E_{\tau-1})$ and two vertices $s,z\in V$, we ask whether there exists a vertex $x\in V$ such that all the possible temporal
paths from $s$ to $z$ visits the vertex $x$. A temporal path from $s$ to $z$ is a sequence 
$(v_0v_1,t_1)(v_1v_2,t_2)\dots (v_{\ell-1}v_\ell,t_\ell)$ such that for all $1\leq i\leq\ell$, $0\leq t_i<\tau$, $t_{i-1}\leq t_i$ and $v_{i-1}v_{i}\in E_{t_i}$.
We say that two edges $e$ and $e'$ are \emph{time-respecting adjacent} if there are $u,v,w\in V$, $e=uv\in E_t$, $e'=vw\in E_{t'}$ and $t\leq t'$. Note that this is not a symmetric relation, due to 
elapsing of time. 
We define $\timeadj(e,e'):= \exists w\exists w'.~\inc(e,w)\wedge \inc(e',w')\wedge \sucstar(w,w')$ to express the fact that $e$ and $e'$ are time-respecting adjacent.

If we want to express that $X$ is closed by time-respecting adjacency restricted to a set $Y$, meaning that we require that if an edge belongs to $X\cap Y$,
any time-respecting adjacent edge belonging to $Y$ should be included in $X$, we use $\closedtempsuc(X,Y):=\forall u\forall v.~ (u\in X\wedge u\in Y\wedge\timeadj(u,v)\wedge v\in Y)\rightarrow v\in X$. This allows to express $\temppath(Y,x,y):=\exists e_1\exists e_2.~ (\inc(e_1,x)\wedge \inc(e_2,y)\wedge e_1\in Y\wedge e_2\in Y \wedge \forall X.~ (e_1\in X\wedge
\closedtempsuc(X,Y)\rightarrow e_2\in X))$, meaning that there is a temporal path from $x$ to $y$ that visits only edges in $Y$. 

Given additional constants $s$ and $z$, we define the following \MSOtwo{} formula:
$cut_{s,z}:=\exists c \forall x\forall y\forall Y.~(\app(s,x)\wedge \app(z,y)\wedge \temppath(Y,x,y))\rightarrow (\exists c' \exists e.~(\app(c,c')\wedge \inc(e,c')\wedge e\in Y))$. This formula states the existence
of a vertex $c$ such that every time there exists a temporal path from the vertex $s$ to the vertex $z$, $c$ appears on this path.

\paragraph*{$\FOone^\Delta$ expression for temporal cliques}
In $\FOone^\Delta$, we can express for instance the existence of a $\Delta$-clique of size $k$ in a temporal graph, for a given $\Delta$. A set of vertices is a $\Delta$-clique if there exists a point in time where every pair of its elements is connected by
an edge \emph{within a timeframe of $\Delta$-snapshots}. This can be done with the formula 
\begin{align*}
\varphi_{\textrm{clique}}&:=\exists x_{11}\dots\exists x_{1\Delta}\exists x_{21}\dots\exists x_{2\Delta}\dots\exists x_{k1}\dots\exists x_{k\Delta}.~ \bigwedge_{(i,j)\neq (i',j')}  \neg (x_{ij}=x_{i'j'})
\wedge\\
& \bigwedge_{1\leq i \leq k} \bigwedge_{1\leq m<n \leq \Delta}\simd_{\Delta}(x_{im}, x_{in})\wedge \bigwedge_{1\leq i<j\leq k}\bigvee_{1\leq m \leq \Delta} \adj(x_{im},x_{jm}).
\end{align*}

\section{Conclusion}
\label{section.conclusion}
We obtain the first two algorithmic meta-theorems for temporal graphs that do not take the lifetime as a parameter: every problem expressible in Monadic Second Order logic, \resp an expressive fragment of First Order logic, over temporal graphs can be solved in polynomial time on temporal graphs whose expanded tree-width, \resp expanded twin-width, is bounded. 
We also initiate the study of temporal graph differentials defined using a local time window of a given size.
To exemplify the usefulness of their introduction, we show another algorithmic meta-theorem for temporal graphs whose differentials have bounded twin-width (or bounded tree-width), for properties given in the same above fragment of First Order logic.
We believe it is an important question to model temporality via differentials:
How to quantify their growth (second derivative)?
How about geometric cases of temporal graphs defined by intersection models in a similar way as unit disk graphs?
In this case, would the positions of the vertex set's embedding help with, or be detrimental to, the definition of derivative?
Another line of research would be to enrich a bit the logic. What predicates can we add that would keep FPT algorithms, \textit{e.g.} in the spirit of~\cite{KKMT19}?

\bibliography{references.bib}

\end{document}